\documentclass{amsart}

\usepackage{mathrsfs,amssymb,setspace}
\usepackage{verbatim, enumerate, color, stmaryrd, bbold, bm}

\theoremstyle{plain}
\newtheorem{theorem}{Theorem}[section]
\newtheorem{lemma}[theorem]{Lemma}
\newtheorem{proposition}[theorem]{Proposition}
\newtheorem{corollary}[theorem]{Corollary}

\theoremstyle{definition}
\newtheorem{notation}[theorem]{Notation}
\newtheorem{example}[theorem]{Example}
\newtheorem{definition}[theorem]{Definition}

\theoremstyle{remark}
\newtheorem{remark}[theorem]{Remark}

\usepackage[all]{xy}

\newcommand{\lref}[1]{Lemma \ref{#1}}

\newcommand{\pref}[1]{Proposition \ref{#1}}
\newcommand{\cref}[1]{Corollary \ref{#1}}
\newcommand{\rref}[1]{Remark \ref{#1}}

\newcommand{\eref}[1]{Example \ref{#1}}
\newcommand{\sref}[1]{Section \ref{#1}}

\newcommand{\powerset}{\raisebox{.15\baselineskip}{\Large\ensuremath{\wp}}}

\begin{document}


\title[{\tt if-then-else} over monoids of non-halting programs and tests]{Axiomatization of {\tt if-then-else} over monoids of possibly non-halting programs and tests}
\author[Gayatri Panicker]{Gayatri Panicker}
\address{Department of Mathematics,  Indian Institute of Technology Guwahati, Guwahati, India}
\email{p.gayatri@iitg.ac.in}
\author[K. V. Krishna]{K. V. Krishna}
\address{Department of Mathematics, Indian Institute of Technology Guwahati, Guwahati, India}
\email{kvk@iitg.ac.in}
\author[Purandar Bhaduri]{Purandar
 Bhaduri}
\address{Department of Computer Science and Engineering, Indian Institute of Technology Guwahati, Guwahati, India}
\email{pbhaduri@iitg.ac.in}


\begin{abstract}
In order to study the axiomatization of the {\tt if-then-else} construct over possibly non-halting programs and tests, the notion of \emph{$C$-sets} was introduced in the literature by considering the tests from an abstract $C$-algebra. This paper extends the notion of $C$-sets to $C$-monoids which include the composition of programs as well as composition of programs with tests. For the class of $C$-monoids where the $C$-algebras are adas a canonical representation in terms of functional $C$-monoids is obtained.
\end{abstract}



\maketitle

\section*{Introduction}

The algebraic properties of the program construct {\tt if-then-else} have been studied in great detail under various contexts. For example, in    \cite{igarashi71,mccarthy63,sethi78}, the authors investigated on axiom schema for determination of the semantic equivalence between the conditional expressions.  The authors in \cite{bloom83,guessarian87,mekler87} studied complete proof systems for various versions of {\tt if-then-else}.  While a transformational characterization of {\tt if-then-else} was given in \cite{manes90}, an axiomatization of equality test algebras was considered in \cite{jackson01,pigozzi91}.  In \cite{bergman91,stokes98}, {\tt if-then-else} was studied as an action of Boolean algebra on a set. Due to their close relation with program features, functions have been canonical models for studies on algebraic semantics of programs.

In \cite{kennison81} Kennison defined comparison algebras as those equipped with a quaternary operation $C(s, t, u, v)$ satisfying certain identities modelling the equality test. He also showed that such algebras are simple if and only if $C$ is the direct comparison operation $C_{0}$ given by $C_{0}(s, t, u, v)$ taking value $u$ if $s = t$ and $v$ otherwise. This was extended by Stokes in \cite{stokes10} to semigroups and monoids. He showed that every comparison semigroup (monoid) is embeddable in  the comparison semigroup (monoid) $\mathcal{T}(X)$ of all total functions $X \rightarrow X$, for some set $X$. He also obtained a similar result in terms of partial functions $X \rightarrow X$. In \cite{jackson09} Jackson and Stokes gave a complete axiomatization of {\tt if-then-else} over halting programs and tests. They also modelled composition of functions and of functions with predicates and called this object a $B$-monoid and further showed that the more natural setting of only considering composition of functions would not admit a finite axiomatization. They proved that every $B$-monoid is embeddable in a functional $B$-monoid comprising total functions and halting tests and thus achieved a Cayley-type theorem for the class of $B$-monoids.
The work listed above predominantly considered the case where the tests are halting and drawn from a Boolean algebra. A natural interest is to study non-halting tests and programs.
%
%

There are multiple studies (e.g., see \cite{bochvar38,heyting34,kleene38,lukasiewicz20}) on extending two-valued Boolean logic to three-valued logic. However McCarthy's logic (cf. \cite{mccarthy63}) is distinct in that it models the short-circuit evaluation exhibited by programming languages that evaluate expressions in sequential order, from left to right. In \cite{guzman90} Guzm\'{a}n and Squier gave a complete axiomatization of McCarthy's three-valued logic and called the corresponding algebra a $C$-algebra, or the algebra of conditional logic. While studying {\tt if-then-else} algebras in \cite{manes93}, Manes defined an {\em ada} (Algebra of Disjoint Alternatives) which is essentially a $C$-algebra equipped with an oracle for the halting problem.

Jackson and Stokes in \cite{jackson15} studied the algebraic theory of computable functions, which can be viewed as possibly non-halting programs, together with composition, {\tt if-then-else} and {\tt while-do}. In this work they assumed that the tests form a Boolean algebra. Further, they demonstrated how an algebra of non-halting tests could be constructed from Boolean tests in their setting. Jackson and Stokes proposed an alternative approach by considering an abstract collection of non-halting tests and posed the following problem:

{\em Characterize the algebras of computable functions associated with an abstract $C$-algebra of non-halting tests.}

The authors in \cite{panicker16} have approached the problem by adopting the approach of Jackson and Stokes in \cite{jackson09}. The notion of a {\em $C$-set} was introduced through which a complete axiomatization for {\tt if-then-else} over a class of possibly non-halting programs and tests, where tests are drawn from an ada, was provided.

In this paper, following the approach of Jackson and Stokes in \cite{jackson09}, we extend the notion of $C$-sets to include composition of possibly non-halting programs and of these programs with possibly non-halting tests. This object is termed a {\em $C$-monoid} and we show that every $C$-monoid where the tests are drawn from an ada is embeddable in a canonical model of $C$-monoids, viz., functional $C$-monoids. The organisation of the paper is as follows. In \sref{SectionPrelims} we provide necessary background material including the notion of $C$-sets and their properties. We introduce the notion of a {\em $C$-monoid} in \sref{SectionCMonoids} and give various examples thereof. In \sref{SectionRepresentationCMonoids} we delineate the procedure to achieve a Cayley-type theorem and embed every $C$-monoid where the tests are drawn from an ada into a functional $C$-monoid. We conclude the paper in \sref{Conclusion}.

\section{Preliminaries} \label{SectionPrelims}


In this section we present the necessary background material. First we recall the concept of $B$-sets. The notion of a $B$-set was introduced by Bergman in \cite{bergman91} and elucidated by Jackson and Stokes in \cite{jackson09} to study the theory of halting programs equipped with the operation of {\tt if-then-else}.

\begin{definition}
 Let $\langle Q, \vee, \wedge, \neg, T, F \rangle$ be a Boolean algebra and $S$
be a set. A \emph{$B$-set} is a pair $(S, Q)$, equipped with a function $\eta : Q \times S \times S \rightarrow
S$, called \emph{$B$-action}, where $\eta(\alpha, a, b)$ is denoted by $\alpha[a, b]$, read ``{\tt if $\alpha$ then $a$ else $b$}", that satisfies the following axioms for all $\alpha, \beta \in Q$ and $a, b, c \in S$:

\begin{align}
  \alpha[a, a] & = a \label{B1} \\
   \alpha[\alpha[a, b], c] & = \alpha[a, c] \label{B2} \\
   \alpha[a, \alpha[b, c]] & = \alpha[a, c] \label{B3} \\
   F[a, b] & = b \label{B4} \\
   \neg \alpha[a, b] & = \alpha[b, a] \label{B5} \\
   (\alpha \wedge \beta)[a, b] & = \alpha[\beta[a, b], b] \label{B6}
\end{align}
\end{definition}

In \cite{jackson09} Jackson and Stokes also considered the case of modelling {\tt if-then-else} over a collection of programs with composition by including an operation to capture the composition of programs with tests.

\begin{definition}
Let $(S, \cdot)$ be a monoid with identity element 1 and $(S, Q)$ be a $B$-set. The pair $(S, Q)$ equipped with a function  $\circ : S \times Q \rightarrow Q$ is said to be a \emph{$B$-monoid} if it satisfies the following axioms for all $a, b, c \in S$ and $\alpha, \beta \in Q$:

\begin{align}
a \circ T & = T \label{BM1} \\
(a \circ \alpha) \wedge (a \circ \beta) & = a \circ (\alpha \wedge \beta) \label{BM2} \\
a \circ (\neg \alpha) & = \neg (a \circ \alpha) \label{BM3} \\
a \circ (b \circ \alpha) & = (a \cdot b) \circ \alpha \label{BM4} \\
\alpha[a, b] \cdot c & = \alpha [a \cdot c, b \cdot c] \label{BM5} \\
a \cdot \alpha[b, c] & = (a \circ \alpha) [a \cdot b, a \cdot c] \label{BM6} \\
\beta[a, b] \circ \alpha & = (\beta \wedge (a \circ \alpha)) \vee (\neg \beta \wedge (b \circ \alpha))\;\; (\mbox{written}\;  \beta \llbracket a \circ \alpha, b \circ \alpha\rrbracket) \label{BM7} \\
1 \circ \alpha & = \alpha \label{BM8}
\end{align}
\end{definition}

Since functions on sets model programs, standard examples of $B$-sets and $B$-monoids come from functions on sets. Let $X$ and $Y$ be two sets. The set of all functions $X \rightarrow Y$ will be denoted by $Y^{X}$ while the set of all functions $X \rightarrow X$ will be denoted by $\mathcal{T}(X)$.

Let $\mathbb{2}$ be the two element Boolean algebra. For any set $X$ the pair $(\mathcal{T}(X), \mathbb{2}^{X})$ is a $B$-set with the following action for all $\alpha \in \mathbb{2}^{X}$ and $g, h \in \mathcal{T}(X)$:
  \begin{equation*}
  \alpha[g, h](x) =
  \begin{cases}
   g(x), & \text{ if } \alpha(x) = T; \\
   h(x), & \text{ if } \alpha(x) = F.
  \end{cases}
 \end{equation*}
Note that $\mathcal{T}(X)$ is a monoid with respect to usual composition of mappings. The $B$-set $(\mathcal{T}(X), \mathbb{2}^{X})$ equipped with the operation $\circ$ defined by
$$(f \circ \alpha) = \{ x \in X : f(x) \in \alpha \}$$
for all $f \in \mathcal{T}(X)$ and $\alpha \in \mathbb{2}^{X}$ is a $B$-monoid. In fact, Jackson and Stokes obtained the following theorem.

\begin{theorem}[\cite{jackson09}] \label{FunctionalBmonoidEmb}
Every $B$-monoid $(S, Q)$ is embeddable as a two-sorted algebra into the $B$-monoid $(\mathcal{T}(X), \mathbb{2}^{X})$ for some set $X$. Furthermore, if $S$ and $Q$ are finite, then $X$ is finite.
\end{theorem}

In \cite{guzman90} Guzm\'{a}n and Squier introduced the notion of a $C$-algebra as the algebra corresponding to McCarthy's three-valued logic (cf. \cite{mccarthy63}).

\begin{definition}
 A \emph{$C$-algebra} is an algebra $\langle M, \vee, \wedge, \neg \rangle$ of
type $(2, 2, 1)$, which satisfies the following axioms for all $\alpha, \beta, \gamma \in M$:

\begin{align}
  \neg \neg \alpha & = \alpha \label{C1} \\
   \neg (\alpha \wedge \beta) & = \neg \alpha \vee \neg \beta \label{C2} \\
   (\alpha \wedge \beta) \wedge \gamma & = \alpha \wedge (\beta \wedge \gamma) \label{C3} \\
   \alpha \wedge (\beta \vee \gamma) & = (\alpha \wedge \beta) \vee (\alpha \wedge \gamma) \label{C4} \\
   (\alpha \vee \beta) \wedge \gamma & = (\alpha \wedge \gamma) \vee (\neg \alpha \wedge \beta \wedge \gamma) \label{C5} \\
   \alpha \vee (\alpha \wedge \beta) & = \alpha \label{C6} \\
   (\alpha \wedge \beta) \vee (\beta \wedge \alpha) & = (\beta \wedge \alpha) \vee (\alpha \wedge \beta) \label{C7}
\end{align}
\end{definition}

It is easy to see that every Boolean algebra is a $C$-algebra. In particular, $\mathbb{2}$ is a $C$-algebra. Let $\mathbb{3}$ denote the $C$-algebra with the universe $\{ T, F, U \}$ and the following operations.
 \begin{center}
  \begin{tabular}{c|c}
  $\neg$ & \\
  \hline
  $T$ & $F$ \\
  $F$ & $T$ \\
  $U$ & $U$
 \end{tabular}
 \quad
 \begin{tabular}{c|ccc}
  $\wedge$ & $T$ & $F$ & $U$ \\
  \hline
  $T$ & $T$ & $F$ & $U$ \\
  $F$ & $F$ & $F$ & $F$ \\
  $U$ & $U$ & $U$ & $U$
 \end{tabular}
 \quad
 \begin{tabular}{c|ccc}
  $\vee$ & $T$ & $F$ & $U$ \\
  \hline
  $T$ & $T$ & $T$ & $T$ \\
  $F$ & $T$ & $F$ & $U$ \\
  $U$ & $U$ & $U$ & $U$
 \end{tabular}
 \end{center}
In fact, the $C$-algebra $\mathbb{3}$ is the McCarthy's three-valued logic.

In view of the fact that the class of $C$-algebras is a variety, for any set $X$, $\mathbb{3}^{X}$ is a $C$-algebra with the operations defined pointwise. Guzm\'{a}n and Squier in \cite{guzman90} showed that elements of $\mathbb{3}^{X}$ along with the $C$-algebra operations may be viewed in terms of \emph{pairs of sets}. This is a pair $(A, B)$ where $A, B \subseteq X$ and $A \cap B = \emptyset$. Akin to the well-known correlation between $\mathbb{2}^{X}$ and the power set $\powerset(X)$ of $X$, for any element $\alpha \in \mathbb{3}^{X}$, associate the pair of sets $(\alpha^{-1}(T), \alpha^{-1}(F))$. Conversely, for any pair of sets $(A, B)$ where $A, B \subseteq X$ and $A \cap B = \emptyset$ associate the function $\alpha$ where $\alpha(x) = T$ if $x \in A$, $\alpha(x) = F$ if $x \in B$ and $\alpha(x) = U$ otherwise. With this correlation, the operations can be expressed as follows:
\begin{align*}
 \neg (A_{1}, A_{2}) & = (A_{2}, A_{1}) \\
 (A_{1}, A_{2}) \wedge (B_{1}, B_{2}) & = (A_{1} \cap B_{1}, A_{2} \cup (A_{1} \cap B_{2})) \\
 (A_{1}, A_{2}) \vee (B_{1}, B_{2}) & = ((A_{1} \cup (A_{2} \cap B_{1}), A_{2} \cap B_{2}) \\
\end{align*}

\begin{notation}
We use $M$ to denote an arbitrary $C$-algebra. By a $C$-algebra with $T, F, U$ we mean a $C$-algebra with nullary operations $T, F, U$, where $T$ is the (unique) left-identity (and right-identity) for $\wedge$, $F$ is the (unique) left-identity (and right-identity) for $\vee$ and $U$ is the (unique) fixed point for $\neg$. Note that $U$ is also a left-zero for both $\wedge$ and $\vee$ while $F$ is a left-zero for $\wedge$.
\end{notation}

In \cite{manes93} Manes introduced the notion of ada (algebra of disjoint alternatives) which is a $C$-algebra equipped with an oracle for the halting problem.

\begin{definition}
 An \emph{ada} is a $C$-algebra $M$ with $T, F, U$ equipped with an additional unary operation $(\text{ })^{\downarrow}$ subject to the following equations for all $\alpha, \beta \in M$:
 \begin{align}
  F^{\downarrow} & = F \label{A1} \\
  U^{\downarrow} & = F \label{A2} \\
  T^{\downarrow} & = T \label{A3} \\
  \alpha \wedge \beta^{\downarrow} & = \alpha \wedge (\alpha \wedge \beta)^{\downarrow} \label{A4} \\
  \alpha^{\downarrow} \vee \neg (\alpha^{\downarrow}) & = T \label{A5} \\
  \alpha & = \alpha^{\downarrow} \vee \alpha \label{A6}
 \end{align}
\end{definition}

The $C$-algebra $\mathbb{3}$ with the unary operation $(\text{ })^{\downarrow}$ defined by \eqref{A1}, \eqref{A2} and \eqref{A3} forms an ada. This ada will also be denoted by $\mathbb{3}$. One may easily resolve the notation overloading -- whether $\mathbb{3}$ is a $C$-algebra or an ada --  depending on the context. In \cite{manes93} Manes showed that the ada $\mathbb{3}$ is the only subdirectly irreducible ada. For any set $X$, $\mathbb{3}^{X}$ is an ada with operations defined pointwise. Note that the ada $\mathbb{3}$ is also simple.

We use the following notations related to sets and equivalence relations.
\begin{notation}$\;$
\begin{enumerate}
\item Let $X$ be a set and $\bot \notin X$. The pointed set $X \cup \{ \bot \}$ with base point $\bot$ is denoted by $X_{\bot}$.
\item The set of all functions on $X_{\bot}$ which fix $\bot$ is denoted by $\mathcal{T}_{o}(X_{\bot})$, i.e., $\mathcal{T}_{o}(X_{\bot}) = \{f \in \mathcal{T}(X_{\bot}) \; : \; f(\bot) = \bot\}$.
\item Under an equivalence relation $\sigma$ on a set $A$, the equivalence class of an element $p \in A$ will be denoted by $\overline{p}^{{\sigma}}$. Within a given context, if there is no ambiguity, we may simply denote the equivalence class by $\overline{p}$.
\end{enumerate}
\end{notation}

In order to axiomatize {\tt if-then-else} over possibly non-halting programmes and tests, in \cite{panicker16}, Panicker et al. considered the tests from a $C$-algebra and introduced the notion of $C$-sets. We now recall the notion of a $C$-set.
\begin{definition}
 Let $S_{\bot}$ be a pointed set with base point $\bot$ and $M$ be a $C$-algebra with $T, F, U$. The pair $(S_{\bot}, M)$  equipped with an action \[\_\; [\_\; , \_] : M \times S_{\bot} \times S_{\bot} \rightarrow S_{\bot}\] is called a \emph{$C$-set} if it satisfies the following axioms for all $\alpha, \beta \in M$ and $s, t, u, v \in S_{\bot}$:
 \begin{align}
  U[s, t] & = \bot \label{EC1} & \text{($U$-axiom)} \\
  F[s, t] & = t \label{EC6} & \text{($F$-axiom)} \\
  (\neg \alpha)[s, t] & = \alpha[t, s] \label{EC5} & \text{($\neg$-axiom)} \\
  \alpha[\alpha[s, t], u] & = \alpha[s, u]  \label{EC3} & \text{(positive redundancy)} \\
  \alpha[s, \alpha[t, u]] & = \alpha[s, u] \label{EC4} & \text{(negative redundancy)} \\
  (\alpha \wedge \beta)[s, t] & = \alpha[\beta[s, t], t] \label{EC7} & \text{($\wedge$-axiom)} \\
  \alpha[\beta[s, t], \beta[u, v]] & = \beta[\alpha[s, u], \alpha[t, v]] \label{EC2} & \text{(premise interchange)} \\
  \alpha[s, t] = \alpha[t, t] & \Rightarrow (\alpha \wedge \beta)[s, t] = (\alpha \wedge \beta)[t, t] \label{EC8} & \text{($\wedge$-compatibility)}
 \end{align}
\end{definition}

Let $M$ be a $C$-algebra with $T, F, U$  treated as a pointed set with base point $U$. The pair $(M, M)$ is a $C$-set under the following action for all $\alpha, \beta, \gamma \in M$:
 $$\alpha[\beta, \gamma] = (\alpha \wedge \beta) \vee (\neg \alpha \wedge \gamma).$$
We denote the action of the $C$-set $(M, M)$ by $\_\; \llbracket \_ \;, \_ \rrbracket$. In \cite{panicker16}, Panicker et al. showed that the axiomatization is complete for the class of $C$-sets $(S_{\bot}, M)$ when $M$ is an ada.  In that connection, they obtained some properties of $C$-sets. Amongst, in \pref{PropCollection} below, we list certain properties related to congruences which are useful in the present work. Viewing $C$-sets as two-sorted algebras,  a \emph{congruence} of a $C$-set is a pair $(\sigma, \tau)$, where $\sigma$ is an equivalence relation on $S_{\bot}$ and $\tau$ is a congruence on the ada $M$ such that $(s, t), (u, v) \in \sigma \; \text{ and }\; (\alpha, \beta) \in \tau$ imply that $(\alpha[s, u], \beta[t, v]) \in \sigma$.

\begin{proposition}[\cite{panicker16}] \label{PropCollection}
Let $(S_{\bot}, M)$ be a $C$-set where $M$ is an ada.  For each maximal congruence $\theta$ on $M$, let $E_\theta$ be the relation on $S_{\bot}$ given by
 $$E_{\theta} = \{ (s, t) \in S_{\bot} \times S_{\bot} : \;  \beta[s, t] = \beta[t, t]\; \text{ for some } \beta \in \overline{T}^{\theta}\}.$$ Then we have the following properties:
 \begin{enumerate}[\rm(i)]
  \item For $\alpha \in M$ and $s, t \in S_{\bot}$, if $(\alpha, \beta) \in \theta$ then according to $\beta = T, F$ or $U$, we have $(\alpha[s, t], s) \in E_{\theta}, (\alpha[s, t], t) \in E_{\theta}$ or $(\alpha[s, t], \bot) \in E_{\theta}$, respectively.
  \item The pair $(E_{\theta}, \theta)$ is a $C$-set congruence.
  \item For the $C$-set $(M, M)$ the equivalence $E_{\theta}$ on $M$,  denoted by $E_{\theta_{M}}$, is a subset of $\theta$.
  \item $\displaystyle \bigcap_{\theta} E_{\theta} = \Delta_{S_{\bot}}$, where $\theta$ ranges over all maximal congruences on $M$.
  \item The intersection of all maximal congruences on $M$ is trivial, that is $\bigcap \theta = \Delta_{M}$ where $\theta$ ranges over all maximal congruences on $M$.
 \end{enumerate}
\end{proposition}
\noindent For more details on $C$-sets one may refer to \cite{panicker16}.


\section{$C$-monoids} \label{SectionCMonoids}

We now include the case where the composition of two elements of the base set and of an element with a predicate is allowed. Our motivating example is $(\mathcal{T}_{o}(X_{\bot}), \mathbb{3}^{X})$, where $\mathcal{T}_{o}(X_{\bot})$ is considered to be a monoid with zero by equipping it with composition of functions. The composition will be written from left to right, i.e., $(f \cdot g)(x) = g(f(x))$. The monoid identity in $\mathcal{T}_{o}(X_{\bot})$ is the identity function $id_{X_{\bot}}$ and the zero element is $\zeta_{\bot}$, the constant function taking the value $\bot$. We also include composition of functions with predicates via the natural interpretation given by the following for all $f \in \mathcal{T}_{o}(X_{\bot})$ and $\alpha \in \mathbb{3}^{X}$:
\begin{equation} \label{FunctionalCMonoidEq}
(f \circ \alpha)(x) =
\begin{cases}
T, & \text{ if } \alpha(f(x)) = T; \\
F, & \text{ if } \alpha(f(x)) = F; \\
U, & \text{ otherwise. }
\end{cases}
\end{equation}

Note that if the composition takes value $T$ or $F$ at some point $x \in X_{\bot}$ then as $\alpha \in \mathbb{3}^{X}$ this implies that $f(x) \neq \bot$.

With this example in mind we define a $C$-monoid as follows.

\begin{definition}
 Let $(S_{\bot}, \cdot)$ be a monoid with identity element $1$ and zero element $\bot$ where $\bot \cdot s = \bot = s \cdot \bot$. Let $M$ be a $C$-algebra and $(S_{\bot}, M)$ be a $C$-set with $\bot$ as the base point of the pointed set $S_{\bot}$. The pair $(S_{\bot}, M)$ equipped with a function $$\circ : S_{\bot} \times M \rightarrow M$$ is said to be a \emph{$C$-monoid} if it satisfies the following axioms for all $s, t, r, u \in S_{\bot}$ and $\alpha, \beta \in M$:
 \begin{align}
 \bot \circ \alpha & = U \label{EM7} & \text{($\bot$-$\circ$-axiom)} \\
 t \circ U & = U \label{EM8} & \text{($U$-$\circ$-axiom)} \\
  1 \circ \alpha & = \alpha \label{EM1} & \text{($1$-$\circ$-axiom)} \\
  s \circ (\neg \alpha) & = \neg (s \circ \alpha) \label{EM4} & \text{($\neg$-$\circ$-axiom)} \\
  s \circ (\alpha \wedge \beta) & = (s \circ \alpha) \wedge (s \circ \beta) \label{EM3} & \text{($\wedge$-$\circ$-axiom)} \\
  (s \cdot t) \circ \alpha & = s \circ (t \circ \alpha) \label{EM2} & \text{(semigroup action)} \\
  \alpha[s, t] \cdot u & = \alpha[s \cdot u, t \cdot u] \label{EM5} & \text{(right composition)} \\
  r \cdot \alpha[s, t] & = (r \circ \alpha)[r \cdot s, r \cdot t] \label{EM9} & \text{(left composition)} \\
  \alpha[s, t] \circ \beta & = \alpha \llbracket s \circ \beta, t \circ \beta \rrbracket \label{EM6} & \text{($\circ$-interchange)}
 \end{align}
\end{definition}

The following are examples of $C$-monoids.

\begin{example} \label{ExampleFunctionalCmonoid}
 Recall from \cite{panicker16} that the pair $\big( \mathcal{T}_{o}(X_{\bot}), \mathbb{3}^{X}\big)$ equipped with the action \eqref{FunctionalAction} for all $f, g \in \mathcal{T}_{o}(X_{\bot})$ and $\alpha \in \mathbb{3}^{X}$ is a $C$-set. Note that $\mathcal{T}_{o}(X_{\bot})$ is treated as a pointed set with base point $\zeta_{\bot}$.

 \begin{equation} \label{FunctionalAction}
  \alpha[f, g](x) =
  \begin{cases}
   f(x), & \text{ if } \alpha(x) = T; \\
   g(x), & \text{ if } \alpha(x) = F; \\
   \bot, & \text{ otherwise. }
  \end{cases}
 \end{equation}

 The $C$-set $\big( \mathcal{T}_{o}(X_{\bot}), \mathbb{3}^{X} \big)$ equipped with the operation $\circ$ given in \eqref{FunctionalCMonoidEq} and with $\mathcal{T}_{o}(X_{\bot})$ treated as a monoid with zero is in fact a $C$-monoid. For verification of axioms \eqref{EM7} -- \eqref{EM6} refer to Appendix \ref{VerFunctionalCmonoid}. Such $C$-monoids will be called \emph{functional $C$-monoids.}
\end{example}

\begin{example} \label{ExampleGeneralCmonoid}
 Let $S_{\bot}$ be a non-trivial monoid with identity $1$ and zero $\bot$ and no non-zero zero-divisors, i.e., $s \cdot t = \bot \Rightarrow s = \bot$ or $t = \bot$. Then $S_{\bot}^{X}$ is also a monoid with zero for any set $X$ with operations defined pointwise. For $f, g \in S_{\bot}^{X}$ define $(f \cdot g)(x) = f(x) \cdot g(x)$. The identity of $S_{\bot}^{X}$ is the constant function $\zeta_{1}$ taking the value $1$. The zero and base point of $S_{\bot}^{X}$ is the constant function $\zeta_{\bot}$ taking the value $\bot$. Recall from \cite{panicker16} that the pair $\big( S_{\bot}^{X}, \mathbb{3}^{X} \big)$ is a $C$-set under action \eqref{FunctionalAction}. In fact it is also a $C$-monoid with $\circ$ defined as follows for all $f \in S_{\bot}^{X}$ and $\alpha \in \mathbb{3}^{X}$:
 \begin{equation*}
  (f \circ \alpha)(x) =
  \begin{cases}
   \alpha(x), & \text{ if } f(x) \neq \bot; \\
   U, & \text{ otherwise.}
  \end{cases}
 \end{equation*}
 For verification of axioms \eqref{EM7} -- \eqref{EM6} refer to Appendix \ref{VerGeneralCmonoid}.
\end{example}

\begin{example} \label{ExampleBasicCmonoid}
 Let $S_{\bot}$ be a non-trivial monoid with zero and no non-zero zero-divisors, i.e., $s \cdot t = \bot \Rightarrow s = \bot \text{ or } t = \bot$. In \cite{panicker16} the authors showed that for any pointed set  $S_{\bot}$ with base point $\bot$, the pair $(S_{\bot}, \mathbb{3})$ is a (basic) $C$-set with respect to the following action for all $a, b \in S_\bot$ and $\alpha \in \mathbb{3}$:
 \begin{equation*}
  \alpha[a, b] =
  \begin{cases}
   a, & \text{ if } \alpha = T; \\
   b, & \text{ if } \alpha = F; \\
   \bot, & \text{ if } \alpha = U.
  \end{cases}
 \end{equation*}
 This basic $C$-set $(S_{\bot}, \mathbb{3})$ equipped with $\circ : S_{\bot} \times \mathbb{3} \rightarrow \mathbb{3}$ defined below for all $s \in S_{\bot}$ and $\alpha \in \mathbb{3}$ is a $C$-monoid.
 \begin{equation*}
  s \circ \alpha =
  \begin{cases}
   \alpha, & \text{ if } s \neq \bot; \\
   U, & \text{ if } s = \bot.
  \end{cases}
 \end{equation*}
 For verification of axioms \eqref{EM7} -- \eqref{EM6} refer to Appendix \ref{VerBasicCmonoid}.
\end{example}

\section{Representation of a class of $C$-monoids} \label{SectionRepresentationCMonoids}

In this section we obtain a Cayley-type theorem for a class of $C$-monoids as stated in the following main theorem.

\begin{theorem}\label{main_theorem}
 Every $C$-monoid $(S_{\bot}, M)$ where $M$ is an ada is embeddable in the $C$-monoid $\big( \mathcal{T}_{o}(X_{\bot}), \mathbb{3}^{X} \big)$ for some set $X$. Moreover, if both $S_{\bot}$ and $M$ are finite then so is $X$.
\end{theorem}

\noindent \emph{Sketch of the proof.} For each maximal congruence $\theta$ of $M$, we consider the $C$-set congruence $(E_\theta, \theta)$ of $(S_{\bot}, M)$. Corresponding to each such congruence, we construct a homomorphism of $C$-monoids from $(S_{\bot}, M)$ to the functional $C$-monoid over the set $S_{\bot} / E_{\theta}$. This collection of homomorphisms has the property that every distinct pair of elements from each component of the $C$-monoid will be separated by some homomorphism from this collection. We then set $X$ to be the disjoint union of $S_{\bot} / E_{\theta}$'s excluding the equivalence class $\overline{\bot}^{_{E_{\theta}}}$. We complete the proof by constructing a  monomorphism -- by pasting together each of the individual homomorphisms from the collection defined earlier -- from the $C$-monoid $(S_{\bot}, M)$ to the functional $C$-monoid over the pointed set $X_{\bot}$ with a new base point $\bot$.

The proof of Theorem \ref{main_theorem} will be developed through various subsections. First in Subsection \ref{prop_max_cong}, we study some properties of maximal congruences of adas. We then present a collection of homomorphisms which separate every distinct pair of elements from each component of $(S_{\bot}, M)$ in Subsection \ref{collec_homos}. In Subsection \ref{embed_func}, we construct the required functional $C$-monoid and establish an embedding from $(S_{\bot}, M)$. Finally, we consolidate the proof in Subsection \ref{proof_main_theorem}.

In what follows $(S_{\bot}, M)$ is a $C$-monoid with $M$ as an ada. Let $\theta$ be a maximal congruence on $M$ and $E_{\theta}$ be the equivalence on $S_\bot$ as defined in \pref{PropCollection} so that the pair $(E_{\theta}, \theta)$ is a congruence on $(S_{\bot}, M)$. We denote the quotient set $S_{\bot} / E_{\theta}$ by $S_{{\theta}_{\bot}}$ and use $S_{\theta}$ to denote the set $S_{{\theta}_{\bot}} \setminus \{ \overline{\bot}^{_{E_{\theta}}} \}$.  Further, we use $q, s, t, u, v$ to denote elements of $S_{\bot}$ and $\alpha, \beta, \gamma$ to denote elements of the ada $M$.

\subsection{Properties of maximal congruences}
\label{prop_max_cong} The following properties are useful in proving the main theorem.

\begin{proposition} \label{PropMaxTheta}
 No two elements of $\{T, F, U\}$ are related under $\theta$. That is,  $(T, F) \notin \theta$, $(T, U) \notin \theta$ and $(F, U) \notin \theta$.
\end{proposition}

\begin{proof}
 If $(T, F) \in \theta$ then we show that $\theta = M \times M$; contradicting the maximality of $\theta$. Suppose $(T, F) \in \theta$ and let $\alpha, \beta \in M$. Then $(T, F), (\alpha, \alpha) \in \theta \Rightarrow (T \wedge \alpha, F \wedge \alpha) \in \theta$ that is $(\alpha, F) \in \theta$. Similarly $(\beta, F) \in \theta$ and so using the symmetry and transitivity of $\theta$ we have $(\alpha, \beta) \in \theta$ and consequently $\theta = M \times M$. The proof of $(T, U) \notin \theta$ follows along similar lines. Finally since $(F, U) \in \theta \Leftrightarrow (T, U) \in \theta$, the result follows.
\end{proof}

\begin{proposition} \label{Prop.Rho.Hom}
 For each $q \in S_{\bot}$, we have
 \begin{enumerate}[\rm(i)]
  \item $(q \circ T)[q, \bot] = q$.
  \item $(q \circ T, F) \notin \theta$.
  \item $(q \circ T, U) \in \theta \Leftrightarrow (q, \bot) \in E_{\theta}$.
  \item $(q \circ T, T) \in \theta \Leftrightarrow (q \circ F, F) \in \theta \Leftrightarrow (q, \bot) \notin E_{\theta}$.
  \item $(s, t) \in E_{\theta} \Rightarrow (s \circ \alpha, t \circ \alpha) \in \theta$ for all $\alpha \in M$.
  \item $(1, \bot) \notin E_{\theta}$.
 \end{enumerate}
\end{proposition}

\begin{proof}
\noindent
 \begin{enumerate}[\rm(i)]
  \item Using \eqref{EM9} we have $q = q \cdot 1 = q \cdot T[1, \bot] = (q \circ T)[q \cdot 1, q \cdot \bot] = (q \circ T)[q, \bot]$.

  \item We prove the result by contradiction. Suppose $(q \circ T, F) \in \theta$. Using the fact that $\theta$ is a congruence on $M$ and \eqref{EM4} we have $(q \circ T, F) \in \theta \Rightarrow (\neg (q \circ T), \neg F) \in \theta \Rightarrow (q \circ (\neg T), \neg F) \in \theta \Rightarrow (q \circ F, T) \in \theta$. Similarly using the fact that $\theta$ is a congruence, \eqref{EM3} and \eqref{EM4} we have $((q \circ F) \vee (q \circ T), (T \vee F)) \in \theta \Rightarrow (q \circ (F \vee T), (T \vee F)) \in \theta \Rightarrow (q \circ T, T) \in \theta$. Thus we have $(q \circ T, F) \in \theta$ and $(q \circ T, T) \in \theta$. From the symmetry and transitivity of $\theta$ it follows that $(T, F) \in \theta$, a contradiction by \pref{PropMaxTheta}. The result follows.

  \item $(\Rightarrow :)$  Let $(q \circ T, U) \in \theta$. Using \pref{PropCollection}(i) we can say that for any choice of $s, t \in S_{\bot}$ we have $((q \circ T)[s, t], \bot) \in E_{\theta}$. On choosing $s = q, t = \bot$ and using \pref{Prop.Rho.Hom}(i) we have $((q \circ T)[q, \bot], \bot) \in E_{\theta}$ that is $(q, \bot) \in E_{\theta}$ as desired.

      $(:\Leftarrow)$ First note that, for $\alpha \in M$,
      \begin{equation}\label{albotisbot}
        \alpha[\bot, \bot] = \bot
      \end{equation}
      (cf. \cite[Proposition 2.8(1)]{panicker16}). Now assume that $(q, \bot) \in E_{\theta}$. Then there exists $\beta \in \overline{T}^{\theta}$ such that $\beta[q, \bot] = \beta[\bot, \bot]$. However, by \eqref{albotisbot}, we have $\beta[q, \bot] = \bot$. Thus, $\beta[q, \bot] \circ T  = \bot \circ T$ so that $\beta \llbracket q \circ T, \bot \circ T \rrbracket  = U$ (using \eqref{EM7} and \eqref{EM6}). Consequently, using \eqref{albotisbot} on $(M, M)$, we have $\beta \llbracket q \circ T, U \rrbracket = \beta \llbracket U, U \rrbracket$.
  Hence $(q \circ T, U) \in E_{\theta_{M}}$ and so from \pref{PropCollection}(iii), $(q \circ T, U) \in \theta$.

  Thus $(q \circ T, U) \in \theta \Leftrightarrow (q, \bot) \in E_{\theta}$.

  \item We first show that $(q \circ T, T) \in \theta \Leftrightarrow (q \circ F, F) \in \theta$ by making use of the substitution property of the congruence $\theta$ with respect to $\neg$, the fact that $\neg$ is an involution and \eqref{EM4}. Thus $(q \circ T, T) \in \theta \Leftrightarrow (\neg (q \circ T), \neg T) \in \theta \Leftrightarrow (q \circ (\neg T), \neg T) \in \theta \Leftrightarrow (q \circ F, F) \in \theta$. Using \pref{PropMaxTheta}, \pref{Prop.Rho.Hom}(ii) and \pref{Prop.Rho.Hom}(iii) we show the equivalence $(q \circ T, T) \in \theta \Leftrightarrow (q, \bot) \notin E_{\theta}$. We have $(q \circ T, T) \in \theta \Rightarrow (q \circ T, U) \notin \theta \Rightarrow (q, \bot) \notin E_{\theta}$. Conversely $(q, \bot) \notin E_{\theta} \Rightarrow (q \circ T, U) \notin \theta$. Using \pref{Prop.Rho.Hom}(ii) it follows that $(q \circ T, F) \notin \theta$. Since $\theta$ is a maximal congruence the only remaining possibility is that $(q \circ T, T) \in \theta$ which completes the proof.

  \item Consider $(s, t) \in E_{\theta}$ and $\alpha \in M$. Then there exists $\beta \in \overline{T}^{\theta}$ such that $\beta[s, t] = \beta[t, t]$. Thus $\beta[s, t] \circ \alpha = \beta[t, t] \circ \alpha$. Using \eqref{EM6} we have $\beta \llbracket s \circ \alpha, t \circ \alpha \rrbracket = \beta \llbracket t \circ \alpha, t \circ \alpha \rrbracket$ from which it follows that $(s \circ \alpha, t \circ \alpha) \in E_{\theta_{M}} \subseteq \theta$ by \pref{PropCollection}(iii).

  \item Suppose that $(1, \bot) \in E_{\theta}$. Using \pref{Prop.Rho.Hom}(v), \eqref{EM1}, \eqref{EM7} we have $(1, \bot) \in E_{\theta} \Rightarrow (1 \circ T, \bot \circ T) \in \theta$ and so $(T, U) \in \theta$ a contradiction by \pref{PropMaxTheta}.
 \end{enumerate}
\end{proof}

\subsection{A class of homomorphisms separating pairs of elements}
\label{collec_homos} For each maximal congruence $\theta$ on $M$, in this subsection, we present homomorphisms $\phi_{\theta} : S_{\bot} \rightarrow \mathcal{T}_{o}(S_{{\theta_{\bot}}})$ and $\rho_{\theta} : M \rightarrow \mathbb{3}^{S_{\theta}}$. Then we establish that $(\phi_{\theta}, \rho_{\theta})$ is a homomorphism from $(S_{\bot}, M)$ to the functional $C$-monoid $\big( \mathcal{T}_{o}(S_{{\theta_{\bot}}}), \mathbb{3}^{S_{\theta}})$. Further, we ascertain that every pair of elements in $S_\bot$ (or $M$) are separated by some $\phi_{\theta}$ (or $\rho_{\theta}$).

\begin{proposition} \label{Prop.phitheta}
The function $\phi_{\theta} : S_{\bot} \rightarrow \mathcal{T}_{o}(S_{{\theta_{\bot}}})$ given by $\phi_{\theta}(s) = \psi_{\theta}^{s}$, where $\psi_{\theta}^{s}(\overline{t}^{_{E_{\theta}}}) = \overline{t \cdot s}^{_{E_{\theta}}}$, is a monoid homomorphism that maps the zero (and base point) of $S_{\bot}$ to that of $\mathcal{T}_{o}(S_{{\theta_{\bot}}})$, that is $\bot \mapsto \zeta_{\overline{\bot}}$.
\end{proposition}

\begin{proof}
\noindent
 \emph{Claim: $\phi_{\theta}$ is well-defined.} It suffices to show that $\psi_{\theta}^{s}$ is well-defined and that $\psi_{\theta}^{s} \in \mathcal{T}_{o}(S_{{\theta_{\bot}}})$, that is $\psi_{\theta}^{s}(\overline{\bot}) = \overline{\bot}$. In order to show the well-definedness of $\psi_{\theta}^{s}$ we consider $\overline{u} = \overline{t}$ that is $(u, t) \in E_{\theta}$. Then there exists $\beta \in \overline{T}^{\theta}$ such that $\beta[u, t] = \beta[t, t]$. Consequently
  \begin{align*}
   \beta[u \cdot s, t \cdot s] & = \beta[u, t] \cdot s & \text{ from } \eqref{EM5} \\
                   & = \beta[t, t] \cdot s & \text{ } \\
                   & = \beta[t \cdot s, t \cdot s] & \text{ from } \eqref{EM5}
  \end{align*}
  Thus $(u \cdot s, t \cdot s) \in E_{\theta}$ and so $\psi_{\theta}^{s}(\overline{u}) = \psi_{\theta}^{s}(\overline{t})$. Also $\psi_{\theta}^{s}(\overline{\bot}) = \overline{\bot \cdot s} = \overline{\bot}$. Thus $\psi_{\theta}^{s} \in \mathcal{T}_{o}(S_{{\theta_{\bot}}})$.

 \emph{Claim: $\phi_{\theta}(\bot) = \zeta_{\overline{\bot}}$\;.} We have $\phi_{\theta}(\bot) = \psi_{\theta}^{\bot}$ where $\psi_{\theta}^{\bot}(\overline{t}) = \overline{t \cdot \bot} = \overline{\bot}$. Thus $\phi_{\theta}(\bot) = \zeta_{\overline{\bot}}$.

 \emph{Claim: $\phi_{\theta}(1) = id_{S_{{\theta}_{\bot}}}$.} We have $\phi_{\theta}(1) = \psi_{\theta}^{1}$ where $\psi_{\theta}^{1}(\overline{t}) = \overline{t \cdot 1} = \overline{t}$.

 \emph{Claim: $\phi_{\theta}$ is a semigroup homomorphism.} Consider $\phi_{\theta}(s \cdot t) = \psi_{\theta}^{s \cdot t}$ where $\psi_{\theta}^{s \cdot t}(\overline{u}) = \overline{u \cdot (s \cdot t)} = \overline{(u \cdot s) \cdot t} = \psi_{\theta}^{t}(\overline{u \cdot s}) = \psi_{\theta}^{t}(\psi_{\theta}^{s}(\overline{u})) = (\psi_{\theta}^{s} \cdot \psi_{\theta}^{t})(\overline{u})$. Thus $\phi_{\theta}(s \cdot t) = \phi_{\theta}(s) \cdot \phi_{\theta}(t)$.
\end{proof}

\begin{proposition} \label{Prop.rhotheta}
 The function $\rho_{\theta} : M \rightarrow \mathbb{3}^{S_{\theta}}$ given by
 \begin{equation*}
 \rho_{\theta}(\alpha) =
 \begin{cases}
  (S_{\theta}, \emptyset), & \text{ if } \alpha = T; \\
  (\emptyset, S_{\theta}), & \text{ if } \alpha = F; \\
  (A_{\theta}^{\alpha}, B_{\theta}^{\alpha}), & \text{ otherwise}
 \end{cases}
 \end{equation*}
 where $A_{\theta}^{\alpha} = \{ \overline{t}^{_{E_{\theta}}} : t \circ \alpha \in \overline{T}^{\theta} \}$ and $B_{\theta}^{\alpha} = \{ \overline{t}^{_{E_{\theta}}} : t \circ \alpha \in \overline{F}^{\theta} \}$, is a homomorphism of $C$-algebras with $T, F, U$.
\end{proposition}

\begin{proof} \emph{Claim: $\rho_{\theta}$ is well-defined.} If $\alpha \in \{ T, F \}$ then the proof is obvious. If $\alpha \notin \{ T, F \}$ then we show that $A_{\theta}^{\alpha} \cap B_{\theta}^{\alpha} = \emptyset$ and that $A_{\theta}^{\alpha}, B_{\theta}^{\alpha} \subseteq S_{\theta}$, that is, $\overline{\bot} \notin A_{\theta}^{\alpha} \cup B_{\theta}^{\alpha}$. Let $\overline{t} \in A_{\theta}^{\alpha} \cap B_{\theta}^{\alpha}$. Then $t \circ \alpha \in \overline{T}^{\theta}$ and $t \circ \alpha \in \overline{F}^{\theta}$ and so $(T, F) \in \theta$ which is a contradiction to \pref{PropMaxTheta}. Using \eqref{EM7} we have $\bot \circ \alpha = U$ and so if $\overline{\bot} \in A_{\theta}^{\alpha} \cup B_{\theta}^{\alpha}$ we would have $\bot \circ \alpha = U \in \{ \overline{T}^{\theta}, \overline{F}^{\theta} \}$, a contradiction to \pref{PropMaxTheta}. Finally we show that the image under $\rho_{\theta}$ is independent of the representative of the equivalence class chosen. Using \pref{Prop.Rho.Hom}(v) we have $\overline{s} = \overline{t} \Rightarrow (s \circ \alpha,  t \circ \alpha) \in \theta$. The result follows.

 \emph{Claim: $\rho_{\theta}$ preserves the constants $T, F, U$.} It is clear that $\rho_{\theta}(T) = (S_{\theta}, \emptyset)$, $\rho_{\theta}(F) = (\emptyset, S_{\theta})$ and, using \eqref{EM8} and \pref{PropMaxTheta}, that $\rho_{\theta}(U) = (A_{\theta}^{U}, B_{\theta}^{U}) = (\emptyset, \emptyset)$ from which the result follows.

 \emph{Claim: $\rho_{\theta}$ is a $C$-algebra homomorphism.} We show that $\rho_{\theta}(\neg \alpha) = \neg (\rho_{\theta}(\alpha)).$ If $\alpha \in \{ T, F \}$ the proof is obvious. Suppose that $\alpha \notin \{ T, F \}$. Then we have the following.
   \begin{align*}
    \rho_{\theta}(\neg \alpha) & = (A_{\theta}^{\neg \alpha}, B_{\theta}^{\neg \alpha}) \\
                               & = (\{ \overline{t} : t \circ (\neg \alpha) \in \overline{T}^{\theta} \}, \{ \overline{t} : t \circ (\neg \alpha) \in \overline{F}^{\theta} \} ) \\
                               & = (\{ \overline{t} : \neg (t \circ \alpha) \in \overline{T}^{\theta} \}, \{ \overline{t} : \neg (t \circ \alpha) \in \overline{F}^{\theta} \} ) & \text{ using } \eqref{EM4} \\
                               & = (\{ \overline{t} : t \circ \alpha \in \overline{F}^{\theta} \}, \{ \overline{t} : t \circ \alpha \in \overline{T}^{\theta} \} ) \\
                               & = (B_{\theta}^{\alpha}, A_{\theta}^{\alpha}) \\
                               & = \neg(\rho_{\theta}(\alpha))
   \end{align*}

 Finally we show that $\rho_{\theta}(\alpha \wedge \beta) = \rho_{\theta}(\alpha) \wedge \rho_{\theta}(\beta)$. Note that the proof of $\rho_{\theta}(\alpha \vee \beta) = \rho_{\theta}(\alpha) \vee \rho_{\theta}(\beta)$ follows using the double negation and De Morgan's laws, viz., \eqref{C1} and \eqref{C2} respectively in conjunction with the fact that $\rho_{\theta}$ preserves $\neg$ and $\wedge$. In order to prove that $\rho_{\theta}(\alpha \wedge \beta) = \rho_{\theta}(\alpha) \wedge \rho_{\theta}(\beta)$ we proceed by considering the following cases. \\

\noindent  \emph{Case I}: $\alpha, \beta \notin \{ T, F \}$. We have the following subcases:

\noindent \emph{Subcase 1}: $\alpha \wedge \beta \notin \{ T, F \}$. Then $\rho_{\theta}(\alpha \wedge \beta) = (A_{\theta}^{\alpha \wedge \beta}, B_{\theta}^{\alpha \wedge \beta})$, $\rho_{\theta}(\alpha) = (A_{\theta}^{\alpha}, B_{\theta}^{\alpha})$ and $\rho_{\theta}(\beta) = (A_{\theta}^{\beta}, B_{\theta}^{\beta})$. Now $(A_{\theta}^{\alpha}, B_{\theta}^{\alpha}) \wedge (A_{\theta}^{\beta}, B_{\theta}^{\beta}) = (A_{\theta}^{\alpha} \cap A_{\theta}^{\beta}, B_{\theta}^{\alpha} \cup (A_{\theta}^{\alpha} \cap B_{\theta}^{\beta}))$. Thus we have to show that
     $$(A_{\theta}^{\alpha \wedge \beta}, B_{\theta}^{\alpha \wedge \beta}) = (A_{\theta}^{\alpha} \cap A_{\theta}^{\beta}, B_{\theta}^{\alpha} \cup (A_{\theta}^{\alpha} \cap B_{\theta}^{\beta})).$$ We show that the pairs of sets are equal componentwise.

     Let $\overline{q} \in A_{\theta}^{\alpha \wedge \beta}$. Then $q \circ (\alpha \wedge \beta) \in \overline{T}^{\theta}$
      \[\begin{array}{ccll}
         &\Rightarrow& ((q \circ \alpha) \wedge (q \circ \beta), T) \in \theta  & \text{(using \eqref{EM3})}\\
         &\Rightarrow& ( (q \circ \alpha) \wedge ((q \circ \alpha) \wedge (q \circ \beta)), (q \circ \alpha) \wedge T) \in \theta & \text{(since $\theta$ is a congruence)}\\
         &\Rightarrow& ((q \circ \alpha) \wedge (q \circ \beta), q \circ \alpha) \in \theta & \text{(using the properties of $\wedge$)}\\
         &\Rightarrow& (q \circ \alpha, T) \in \theta & \text{(by transitivity of $\theta$)}
      \end{array}\]
      so that $\overline{q} \in A_{\theta}^{\alpha}$. Along similar lines one can observe that \[((((q \circ \alpha) \wedge (q \circ \beta)) \wedge (q \circ \beta)), T \wedge (q \circ \beta)) \in \theta.\] Consequently  $(q \circ \beta, T) \in \theta$ so that $\overline{q} \in A_{\theta}^{\beta}$.  Hence $A_{\theta}^{\alpha \wedge \beta} \subseteq A_{\theta}^{\alpha} \cap A_{\theta}^{\beta}$.

      For reverse inclusion let $\overline{q} \in A_{\theta}^{\alpha} \cap A_{\theta}^{\beta}$. Then  $(q \circ \alpha, T), (q \circ \beta, T) \in \theta$. Since $\theta$ is a congruence we have $((q \circ \alpha) \wedge (q \circ \beta), T \wedge T) = ((q \circ \alpha) \wedge (q \circ \beta), T) = ((q \circ (\alpha \wedge \beta), T) \in \theta$ and so $\overline{q} \in A_{\theta}^{\alpha \wedge \beta}$. Hence $A_{\theta}^{\alpha \wedge \beta} = A_{\theta}^{\alpha} \cap A_{\theta}^{\beta}$.

      In order to show that $B_{\theta}^{\alpha \wedge \beta} \subseteq B_{\theta}^{\alpha} \cup (A_{\theta}^{\alpha} \cap B_{\theta}^{\beta})$ consider $\overline{q} \in B_{\theta}^{\alpha \wedge \beta}$ that is $(q \circ (\alpha \wedge \beta), F) \in \theta$. Since $\theta$ is a maximal congruence consider the following three possibilities:
      \begin{description}
       \item[$(q \circ \alpha, F) \in \theta$] Then clearly $\overline{q} \in B_{\theta}^{\alpha}$ and so $\overline{q} \in B_{\theta}^{\alpha} \cup (A_{\theta}^{\alpha} \cap B_{\theta}^{\beta})$. \\
      \item[$(q \circ \alpha, T) \in \theta$] Then we have $\overline{q} \in A_{\theta}^{\alpha}$.  We show that $(q \circ \beta, F) \in \theta$. If this is not the case then either $(q \circ \beta, T) \in \theta$ or $(q \circ \beta, U) \in \theta$. If $(q \circ \beta, T) \in \theta$ then since $(q \circ \alpha, T) \in \theta$ we have $(q \circ (\alpha \wedge \beta), T \wedge T) = (q \circ (\alpha \wedge \beta), T) \in \theta$ using \eqref{EM3} and the fact that $\theta$ is a congruence. However since $(q \circ (\alpha \wedge \beta), F) \in \theta$ we obtain a contradiction that $(T, F) \in \theta$ (cf. \pref{PropMaxTheta}). Along similar lines if $(q \circ \beta, U) \in \theta$ then as $(q \circ \alpha, T) \in \theta$ we have $(q \circ (\alpha \wedge \beta), U) \in \theta$ and so $(F, U) \in \theta$ a contradiction to \pref{PropMaxTheta}. Hence $(q \circ \beta, F) \in \theta$ so that $\overline{q} \in (A_{\theta}^{\alpha} \cap B_{\theta}^{\beta}) \subseteq B_{\theta}^{\alpha} \cup (A_{\theta}^{\alpha} \cap B_{\theta}^{\beta})$. \\
      \item[$(q \circ \alpha, U) \in \theta$] Since  $(q \circ \beta, q \circ \beta) \in \theta$  we have  $(q \circ (\alpha \wedge \beta), U \wedge q \circ \beta) = (q \circ (\alpha \wedge \beta), U) \in \theta$ using  \eqref{EM3} and the fact that $\theta$ is a congruence. However since $(q \circ (\alpha \wedge \beta), F) \in \theta$ we have $(F, U) \in \theta$ a contradiction to \pref{PropMaxTheta}. Thus this case cannot occur.
      \end{description}

      To show the reverse inclusion let $\overline{q} \in B_{\theta}^{\alpha} \cup (A_{\theta}^{\alpha} \cap B_{\theta}^{\beta})$ that is $\overline{q} \in B_{\theta}^{\alpha}$ or $\overline{q} \in A_{\theta}^{\alpha} \cap B_{\theta}^{\beta}$.
      If $\overline{q} \in B_{\theta}^{\alpha}$ then $(q \circ \alpha, F) \in \theta$
      \[\begin{array}{ccll}
         &\Rightarrow& ((q \circ \alpha) \wedge (q \circ \beta), F \wedge (q \circ \beta)) \in \theta & \text{(since $\theta$ is a congruence)}\\
         &\Rightarrow& (q \circ (\alpha \wedge \beta), F \wedge (q \circ \beta)) \in \theta  & \text{(using \eqref{EM3})} \\
         &\Rightarrow& (q \circ (\alpha \wedge \beta), F) \in \theta & \text{(since $F$ is a left-zero for $\wedge$)}\\
      \end{array}\]
      from which it follows that $\overline{q} \in B_{\theta}^{\alpha \wedge \beta}$. In the case where $\overline{q} \in A_{\theta}^{\alpha} \cap B_{\theta}^{\beta}$ that is $(q \circ \alpha, T), (q \circ \beta, F) \in \theta$, along similar lines it follows that $(q \circ (\alpha \wedge \beta), T \wedge F) = (q \circ (\alpha \wedge \beta), F) \in \theta$ and so $\overline{q} \in B_{\theta}^{\alpha \wedge \beta}$. Therefore $B_{\theta}^{\alpha \wedge \beta} = B_{\theta}^{\alpha} \cup (A_{\theta}^{\alpha} \cap B_{\theta}^{\beta})$. \\

\noindent \emph{Subcase 2}: $\alpha \wedge \beta \in \{ T, F \}$.  Using the fact that $M \leq \mathbb{3}^{X}$ for some set $X$ it is easy to see that if $\alpha, \beta \notin \{ T, F \}$ then $\alpha \wedge \beta \neq T$. It follows that the only possibility in this case is that $\alpha \wedge \beta = F$. Therefore $\rho_{\theta}(\alpha \wedge \beta) = \rho_{\theta}(F) = (\emptyset, S_{\theta})$ and $\rho_{\theta}(\alpha) \wedge \rho_{\theta}(\beta) = (A_{\theta}^{\alpha} \cap A_{\theta}^{\beta}, B_{\theta}^{\alpha} \cup (A_{\theta}^{\alpha} \cap B_{\theta}^{\beta}))$ and so we have to show that
$$(\emptyset, S_{\theta}) = (A_{\theta}^{\alpha} \cap A_{\theta}^{\beta}, B_{\theta}^{\alpha} \cup (A_{\theta}^{\alpha} \cap B_{\theta}^{\beta})).$$

      We first show that $A_{\theta}^{\alpha} \cap A_{\theta}^{\beta} = \emptyset$. If $A_{\theta}^{\alpha} \cap A_{\theta}^{\beta} \neq \emptyset$ then let $\overline{q} \in A_{\theta}^{\alpha} \cap A_{\theta}^{\beta}$ so that $(q \circ \alpha, T) \in \theta, (q \circ \beta, T)  \in \theta$
      \[\begin{array}{ccll}
         &\Rightarrow& ((q \circ \alpha) \wedge (q \circ \beta), T \wedge T) = ((q \circ \alpha) \wedge (q \circ \beta), T) \in \theta  & \text{(since $\theta$ is a congruence)} \\
         &\Rightarrow& (q \circ (\alpha \wedge \beta), T) \in \theta  & \text{(using \eqref{EM3})} \\
         &\Rightarrow& (q \circ F, T) \in \theta & \text{(since $\alpha \wedge \beta = F$)}\\
         &\Rightarrow& (\neg (q \circ F), \neg T) = (\neg (q \circ F), F) \in \theta & \text{(since $\theta$ is a congruence)}\\
         &\Rightarrow& (q \circ \neg F, F) = (q \circ T, F) \in \theta & \text{(using \eqref{EM4})}\\
      \end{array}\]
which is a contradiction to \pref{Prop.Rho.Hom}(ii). Hence $A_{\theta}^{\alpha} \cap A_{\theta}^{\beta} = \emptyset$.

      In order to show that $B_{\theta}^{\alpha} \cup (A_{\theta}^{\alpha} \cap B_{\theta}^{\beta}) = S_{\theta}$ consider $\overline{q} \in S_{\theta}$ that is $\overline{q} \neq \overline{\bot}$ which gives $(q, \bot) \notin E_{\theta}$. We proceed by considering the following three cases:
      \begin{description}
       \item[$(q \circ \alpha, F) \in \theta$] Then it is clear that $\overline{q} \in B_{\theta}^{\alpha} \subseteq B_{\theta}^{\alpha} \cup (A_{\theta}^{\alpha} \cap B_{\theta}^{\beta})$. \\
      \item[$(q \circ \alpha, T) \in \theta$] Then we have $\overline{q} \in A_{\theta}^{\alpha}$. We show that $(q \circ \beta, F) \in \theta$. Suppose that this is not the case. Since $\theta$ is a maximal congruence it implies that either $(q \circ \beta, T) \in \theta$ or $(q \circ \beta, U) \in \theta$. If $(q \circ \beta, T) \in \theta$ then since $(q \circ \alpha, T) \in \theta$ it follows that $(q \circ (\alpha \wedge \beta), T \wedge T) = (q \circ F, T) \in \theta$ so that $(q \circ T, F) \in \theta$. This is a contradiction to \pref{Prop.Rho.Hom}(ii). In the case that $(q \circ \beta, U) \in \theta$ proceeding as earlier we have $(q \circ (\alpha \wedge \beta), T \wedge U) = (q \circ F, U) \in \theta$ so that $(q \circ T, U) \in \theta$. It follows from \pref{Prop.Rho.Hom}(iii) that $(q, \bot) \in E_{\theta}$ which is a contradiction to the assumption that $\overline{q} \in S_{\theta}$. Consequently it must be the case that $(q \circ \beta, F) \in \theta$ so that $\overline{q} \in A_{\theta}^{\alpha} \cap B_{\theta}^{\beta} \subseteq B_{\theta}^{\alpha} \cup (A_{\theta}^{\alpha} \cap B_{\theta}^{\beta})$. \\
      \item[$(q \circ \alpha, U) \in \theta$] Since $\theta$ is a congruence we have $(q \circ \beta, q \circ \beta) \in \theta$
      \[\begin{array}{ccll}
         &\Rightarrow& ((q \circ \alpha) \wedge (q \circ \beta), U \wedge (q \circ \beta))  \in \theta & \text{(since $\theta$ is a congruence)}\\
         &\Rightarrow& ((q \circ (\alpha \wedge \beta), U) = (q \circ F, U) \in \theta  & \text{(since $U$ is a left-zero for $\wedge$ and using \eqref{EM3})} \\
         &\Rightarrow& (\neg (q \circ F), \neg U) \in \theta & \text{(since $\theta$ is a congruence)}\\
         &\Rightarrow& (q \circ \neg F, \neg U) = (q \circ T, U) \in \theta & \text{(using \eqref{EM4})}\\
      \end{array}\]
        Thus using \pref{Prop.Rho.Hom}(iii) we have $(q, \bot) \in E_{\theta}$ which is a contradiction to the assumption that $\overline{q}\in S_{\theta}$. Hence this case cannot occur.
      \end{description}
      Thus $B_{\theta}^{\alpha} \cup (A_{\theta}^{\alpha} \cap B_{\theta}^{\beta}) = S_{\theta}$ which completes the proof in the case where $\alpha, \beta \notin \{ T, F \}$. \\

\noindent \emph{Case II}: $\alpha \in \{ T, F \}$. The verification is straightforward by considering $\alpha = T$ and $\alpha = F$ casewise.

\noindent \emph{Subcase 1}: $\alpha = T$. Then $\rho_{\theta}(\alpha \wedge \beta) = \rho_{\theta}(T \wedge \beta) = \rho_{\theta}(\beta) = (S_{\theta}, \emptyset) \wedge \rho_{\theta}(\beta) = \rho_{\theta}(T) \wedge \rho_{\theta}(\beta) = \rho_{\theta}(\alpha) \wedge \rho_{\theta}(\beta)$. \\

\noindent \emph{Subcase 2}: $\alpha = F$. Then $\rho_{\theta}(\alpha \wedge \beta) = \rho_{\theta}(F \wedge \beta) = \rho_{\theta}(F) = (\emptyset, S_{\theta}) = (\emptyset, S_{\theta}) \wedge \rho_{\theta}(\beta) = \rho_{\theta}(F) \wedge \rho_{\theta}(\beta) = \rho_{\theta}(\alpha) \wedge \rho_{\theta}(\beta)$. \\

\noindent \emph{Case III}: $\beta \in \{ T, F \}$. We have the following subcases:

\noindent \emph{Subcase 1}: $\beta = T$. The proof follows along the same lines as \emph{Case II} above since $T$ is the left and right-identity for $\wedge$. Thus $\rho_{\theta}(\alpha \wedge \beta) = \rho_{\theta}(\alpha \wedge T) = \rho_{\theta}(\alpha) = \rho_{\theta}(\alpha) \wedge (S_{\theta}, \emptyset) = \rho_{\theta}(\alpha) \wedge \rho_{\theta}(T) = \rho_{\theta}(\alpha) \wedge \rho_{\theta}(\beta)$. \\

\noindent \emph{Subcase 2}: $\beta = F$. If $\alpha \in \{ T, F \}$ then this reduces to \emph{Case II} proved above and consequently we have $\rho_{\theta}(\alpha \wedge \beta) = \rho_{\theta}(\alpha) \wedge \rho_{\theta}(\beta)$ in this case. Thus it remains to consider the case where $\alpha \notin \{ T, F \}$. We then have the following subcases depending on $\alpha \wedge \beta$:
      \begin{description}
       \item[$\alpha \wedge \beta \notin \{ T, F \}$] Then $\rho_{\theta}(\alpha \wedge \beta) = \rho_{\theta}(\alpha \wedge F) = (A_{\theta}^{\alpha \wedge F}, B_{\theta}^{\alpha \wedge F})$ while $\rho_{\theta}(\alpha) = (A_{\theta}^{\alpha}, B_{\theta}^{\alpha})$ and $\rho_{\theta}(\beta) = \rho_{\theta}(F) = (\emptyset, S_{\theta})$. Thus $\rho_{\theta}(\alpha) \wedge \rho_{\theta}(F) = (\emptyset, A_{\theta}^{\alpha} \cup B_{\theta}^{\alpha})$. We show that $$(A_{\theta}^{\alpha \wedge F}, B_{\theta}^{\alpha \wedge F}) = (\emptyset, A_{\theta}^{\alpha} \cup B_{\theta}^{\alpha})$$ as earlier by proving that the pairs of sets are equal componentwise. \\

           We show that $A_{\theta}^{\alpha \wedge F} = \emptyset$ by contradiction. If $A_{\theta}^{\alpha \wedge F} \neq \emptyset$ then consider $\overline{q} \in A_{\theta}^{\alpha \wedge F}$. It follows that $(q \circ (\alpha \wedge F), T) \in \theta$
                 \[\begin{array}{ccll}
         &\Rightarrow& ((q \circ (\alpha \wedge F)) \wedge (q \circ F), T \wedge q \circ F) \in \theta & \text{(since $\theta$ is a congruence)} \\
         &\Rightarrow& ((q \circ (\alpha \wedge F)) \wedge (q \circ F), q \circ F) \in \theta & \text{(since $T$ is a left-identity for $\wedge$)} \\
         &\Rightarrow& (((q \circ \alpha) \wedge (q \circ F)) \wedge (q \circ F), q \circ F) \in \theta & \text{(using \eqref{EM3})} \\
         &\Rightarrow& ((q \circ \alpha) \wedge (q \circ F), q \circ F) \in \theta & \text{(using the properties of $\wedge$)} \\
         &\Rightarrow& (q \circ F, T) \in \theta & \text{(since $\theta$ is a congruence)} \\
         &\Rightarrow& (q \circ T, F) \in \theta & \text{(from \eqref{EM4} and since $\theta$ is a congruence)}
      \end{array}\]
      which is a contradiction to \pref{Prop.Rho.Hom}(ii). Hence $A_{\theta}^{\alpha \wedge F} = \emptyset$. \\

    We show that $B_{\theta}^{\alpha \wedge F} = A_{\theta}^{\alpha} \cup B_{\theta}^{\alpha}$ using standard set theoretic arguments. Let $\overline{q} \in B_{\theta}^{\alpha \wedge F}$ and so $(q \circ (\alpha \wedge F), F) \in \theta$ so that $((q \circ \alpha) \wedge (q \circ F), F) \in \theta$. In view of the maximality of $\theta$ it suffices to consider three cases. If either $(q \circ \alpha, T) \in \theta$ or $(q \circ \alpha, F) \in \theta$ then $\overline{q} \in A_{\theta}^{\alpha} \cup B_{\theta}^{\alpha}$. If $(q \circ \alpha, U) \in \theta$ then $((q \circ \alpha) \wedge ((q \circ \alpha) \wedge (q \circ F)), U \wedge F) = ((q \circ \alpha) \wedge (q \circ F), U) \in \theta$.Thus $(F, U) \in \theta$ which is a contradiction to \pref{PropMaxTheta}. Hence this case cannot occur and so $B_{\theta}^{\alpha \wedge F} \subseteq A_{\theta}^{\alpha} \cup B_{\theta}^{\alpha}$. \\

    For the reverse inclusion consider $\overline{q} \in A_{\theta}^{\alpha} \cup B_{\theta}^{\alpha}$ so that $\overline{q} \in A_{\theta}^{\alpha}$ or $\overline{q} \in B_{\theta}^{\alpha}$. If $\overline{q} \in A_{\theta}^{\alpha}$ then $(q \circ \alpha, T) \in \theta$. Since $\overline{q} \in A_{\theta}^{\alpha} \subseteq S_{\theta}$ using \pref{Prop.Rho.Hom}(iv) we have $(q, \bot) \notin E_{\theta} \Rightarrow (q \circ F, F) \in \theta$. Consequently $(q \circ (\alpha \wedge F), (T \wedge F)) = (q \circ (\alpha \wedge F), F) \in \theta$ and so $\overline{q} \in B_{\theta}^{\alpha \wedge F}$. Along similar lines if $\overline{q} \in B_{\theta}^{\alpha}$ we have $(q \circ (\alpha \wedge F), F) \in \theta$ so that $\overline{q} \in B_{\theta}^{\alpha \wedge F}$. Hence $(A_{\theta}^{\alpha \wedge F}, B_{\theta}^{\alpha \wedge F}) = (\emptyset, A_{\theta}^{\alpha} \cup B_{\theta}^{\alpha})$. \\

       \item[$\alpha \wedge \beta \in \{ T, F \}$] Using the fact that $M \leq \mathbb{3}^{X}$ for some set $X$ we have $\alpha \wedge F \neq T$ from which it follows that the only case is $\alpha \wedge \beta = \alpha \wedge F = F$. Thus $\rho_{\theta}(\alpha \wedge F) = \rho_{\theta}(F) = (\emptyset, S_{\theta})$ while $\rho_{\theta}(\alpha) \wedge \rho_{\theta}(F) = (\emptyset, A_{\theta}^{\alpha} \cup B_{\theta}^{\alpha})$. We show that $$(\emptyset, S_{\theta}) = (\emptyset, A_{\theta}^{\alpha} \cup B_{\theta}^{\alpha}).$$ In order to show that $A_{\theta}^{\alpha} \cup B_{\theta}^{\alpha} = S_{\theta}$ consider $\overline{q} \in S_{\theta}$. If $(q \circ \alpha, T) \in \theta$ or $(q \circ \alpha, F) \in \theta$ then the proof is complete. If $(q \circ \alpha, U) \in \theta$ then since $\overline{q} \neq \overline{\bot}$ that is $(q, \bot) \notin E_{\theta}$ by \pref{Prop.Rho.Hom}(iv) we have $(q \circ F, F) \in \theta$. Thus  $(q \circ (\alpha \wedge F), U \wedge F) = (q \circ F, U) \in \theta$. Consequently from the transitivity of $\theta$ it follows that $(F, U) \in \theta$ which is a contradiction to \pref{PropMaxTheta}. Hence $(\emptyset, S_{\theta}) = (\emptyset, A_{\theta}^{\alpha} \cup B_{\theta}^{\alpha})$.
      \end{description}
Thus $\rho_{\theta}$ is a homomorphism of $C$-algebras with $T, F, U$.
\end{proof}

\begin{lemma} \label{LemmaPhiRhoTheta}
 The pair $(\phi_{\theta}, \rho_{\theta})$ is a $C$-monoid homomorphism from $(S_{\bot}, M)$ to the functional $C$-monoid $\big( \mathcal{T}_{o}(S_{{\theta_{\bot}}}), \mathbb{3}^{S_{\theta}} \big)$.
\end{lemma}

\begin{proof}
 In view of \pref{Prop.phitheta} and \pref{Prop.rhotheta} it suffices to show that $\phi_{\theta}(\alpha[s, t]) = \rho_{\theta}(\alpha)[\phi_{\theta}(s), \phi_{\theta}(t)]$ and $\rho_{\theta}(s \circ \alpha) = \phi_{\theta}(s) \circ \rho_{\theta}(\alpha)$ hold. In order to show that $\phi_{\theta}(\alpha[s, t]) = \rho_{\theta}(\alpha)[\phi_{\theta}(s), \phi_{\theta}(t)]$ we proceed casewise depending on the value of $\alpha$ as per the following: \\

 \noindent \emph{Case I}: $\alpha \in \{ T, F \}$. If $\alpha = T$ then $\phi_{\theta}(\alpha[s, t]) = \phi_{\theta}(T[s, t]) = \phi_{\theta}(s) = (S_{\theta}, \emptyset)[\phi_{\theta}(s), \phi_{\theta}(t)] = \rho_{\theta}(T)[\phi_{\theta}(s), \phi_{\theta}(t)] = \rho_{\theta}(\alpha)[\phi_{\theta}(s), \phi_{\theta}(t)]$. Along similar lines if $\alpha = F$ then $\phi_{\theta}(\alpha[s, t]) = \phi_{\theta}(F[s, t]) = \phi_{\theta}(t) = (\emptyset, S_{\theta})[\phi_{\theta}(s), \phi_{\theta}(t)] = \rho_{\theta}(F)[\phi_{\theta}(s), \phi_{\theta}(t)] = \rho_{\theta}(\alpha)[\phi_{\theta}(s), \phi_{\theta}(t)]$. \\

 \noindent \emph{Case II}: $\alpha\notin \{ T, F \}$. If $\alpha \notin \{ T, F \}$ then using \eqref{EM9} we have $\phi_{\theta}(\alpha[s, t]) = \psi_{\theta}^{\alpha[s, t]}$ where $\psi_{\theta}^{\alpha[s, t]}(\overline{v}) = \overline{v \cdot (\alpha[s, t])} = \overline{(v \circ \alpha)[v \cdot s, v \cdot t]}$. Consider $\rho_{\theta}(\alpha)[\phi_{\theta}(s), \phi_{\theta}(t)] = (A_{\theta}^{\alpha}, B_{\theta}^{\alpha})[\psi_{\theta}^{s}, \psi_{\theta}^{t}]$, where
 \begin{align*}
 (A_{\theta}^{\alpha}, B_{\theta}^{\alpha})[\psi_{\theta}^{s}, \psi_{\theta}^{t}](\overline{v}) =
 \begin{cases}
  \overline{v \cdot s}, & \text{ if } \overline{v} \in A_{\theta}^{\alpha}, \text{ that is } (v \circ \alpha) \in \overline{T}^{\theta}; \\
  \overline{v \cdot t}, & \text{ if } \overline{v} \in B_{\theta}^{\alpha} \text{ that is } (v \circ \alpha) \in \overline{F}^{\theta}; \\
  \overline{\bot}, & \text{ otherwise.}
 \end{cases}
 \end{align*}
 It suffices to consider the following three cases: \\
 \noindent \emph{Subcase 1}: $(v \circ \alpha) \in \overline{T}^{\theta}$. using \pref{PropCollection}(i) we have $((v \circ \alpha)[v \cdot s, v \cdot t], v \cdot s) \in E_{\theta}$. Consequently $\overline{(v \circ \alpha)[v \cdot s, v \cdot t]} = \overline{v \cdot s}$. \\

 \noindent \emph{Subcase 2}: $(v \circ \alpha) \in \overline{F}^{\theta}$. Along similar lines if $(v \circ \alpha) \in \overline{F}^{\theta}$ then $((v \circ \alpha)[v \cdot s, v \cdot t], v \cdot t) \in E_{\theta}$, by \pref{PropCollection}(i) and so $\overline{(v \circ \alpha)[v \cdot s, v \cdot t]} = \overline{v \cdot t}$. \\

 \noindent \emph{Subcase 3}: $(v \circ \alpha) \in \overline{U}^{\theta}$. Then $((v \circ \alpha)[v \cdot s, v \cdot t], \bot) \in E_{\theta}$, by \pref{PropCollection}(i) which gives $\overline{(v \circ \alpha)[v \cdot s, v \cdot t]} = \overline{\bot}$. \\

 Thus we have $\psi_{\theta}^{\alpha[s, t]}(\overline{v}) = (A_{\theta}^{\alpha}, B_{\theta}^{\alpha})[\psi_{\theta}^{s}, \psi_{\theta}^{t}](\overline{v})$ for every $\overline{v} \in S_{{\theta}_{\bot}}$ and so $\phi_{\theta}(\alpha[s, t]) = \rho_{\theta}(\alpha)[\phi_{\theta}(s), \phi_{\theta}(t)]$. \\

 We show that $\rho_{\theta}(s \circ \alpha) = \phi_{\theta}(s) \circ \rho_{\theta}(\alpha)$ by proceeding casewise depending on the value of $\alpha$ and $s \circ \alpha$. \\

  \noindent \emph{Case I}: $\alpha \notin \{ T, F \}, s \circ \alpha \notin \{ T, F \}$. Then $\rho_{\theta}(s \circ \alpha) = (A_{\theta}^{s \circ \alpha}, B_{\theta}^{s \circ \alpha})$ and $\rho_{\theta}(\alpha) = (A_{\theta}^{\alpha}, B_{\theta}^{\alpha})$. Then $\phi_{\theta}(s) \circ (A_{\theta}^{\alpha}, B_{\theta}^{\alpha}) = \psi_{\theta}^{s} \circ (A_{\theta}^{\alpha}, B_{\theta}^{\alpha}) = (C, D)$, where $C = \{ \overline{q} \in S_{\theta} : \psi_{\theta}^{s}(\overline{q}) \in A_{\theta}^{\alpha} \}$ and $D = \{ \overline{q} \in S_{\theta} : \psi_{\theta}^{s}(\overline{q}) \in B_{\theta}^{\alpha} \}$. We have to show that $$(A_{\theta}^{s \circ \alpha}, B_{\theta}^{s \circ \alpha}) = (C, D).$$ It is clear that $\overline{q} \in C$
          \[\begin{array}{ccll}
         &\Leftrightarrow& \psi_{\theta}^{s}(\overline{q}) \in A_{\theta}^{\alpha} \\
         &\Leftrightarrow& \overline{q \cdot s} \in A_{\theta}^{\alpha} \\
         &\Leftrightarrow& ((q \cdot s) \circ \alpha, T) \in \theta \\
         &\Leftrightarrow& (q \circ (s \circ \alpha), T) \in \theta & \text{(using \eqref{EM2})}\\
         &\Leftrightarrow& \overline{q} \in A_{\theta}^{s \circ \alpha} \\
       \end{array}\]
  Along similar lines we have $\overline{q} \in D \Leftrightarrow \overline{q} \in B_{\theta}^{s \circ \alpha}$. \\

 \noindent \emph{Case II}: $\alpha \in \{ T, F \}, s \circ \alpha \notin \{ T, F \}$. If $\alpha = T$ then $\rho_{\theta}(s \circ \alpha) = \rho_{\theta}(s \circ T) = (A_{\theta}^{s \circ T}, B_{\theta}^{s \circ T})$. On the other hand $\phi_{\theta}(s) \circ \rho_{\theta}(\alpha) = \phi_{\theta}(s) \circ \rho_{\theta}(T) = \psi_{\theta}^{s} \circ (S_{\theta}, \emptyset) = (C, D)$ where $C = \{ \overline{q} \in S_{\theta} : \psi_{\theta}^{s}(\overline{q}) \in S_{\theta} \}$ and $D = \emptyset$. We have to show that $$(A_{\theta}^{s \circ T}, B_{\theta}^{s \circ T}) = (C, \emptyset).$$
 We show that $B_{\theta}^{s \circ T} = \emptyset$ by contradiction. If $B_{\theta}^{s \circ T} \neq \emptyset$ then let $\overline{q} \in B_{\theta}^{s \circ T}$
        \[\begin{array}{ccll}
         &\Rightarrow& (q \circ (s \circ T), F) \in \theta \\
         &\Rightarrow& ((q \cdot s) \circ T, F) \in \theta & \text{(using \eqref{EM2})} \\
       \end{array}\]
 which is a contradiction to \pref{Prop.Rho.Hom}(ii). Thus $B_{\theta}^{s \circ T} = \emptyset$. \\
 We now show that $A_{\theta}^{s \circ T} = C$. It is clear that $\overline{q} \in C$
       \[\begin{array}{ccll}
        &\Leftrightarrow& \psi_{\theta}^{s}(\overline{q}) \in S_{\theta} \\
        &\Leftrightarrow& \overline{q \cdot s} \in S_{\theta} \\
        &\Leftrightarrow& (q \cdot s, \bot) \notin E_{\theta} \\
        &\Leftrightarrow& ((q \cdot s) \circ T, T) \in \theta & \text{(using \pref{Prop.Rho.Hom}(iv))}\\
        &\Leftrightarrow& (q \circ (s \circ T), T) \in \theta & \text{(using \eqref{EM2})}\\
        &\Leftrightarrow& \overline{q} \in A_{\theta}^{s \circ T} \\
       \end{array}\]
 In the case where $\alpha = F$ the proof follows along similar lines. \\

 \noindent \emph{Case III}: $\alpha \notin \{ T, F \}, s \circ \alpha \in \{ T, F \}$. We have the following subcases: \\
 \noindent \emph{Subcase 1}: $s \circ \alpha = T$. Then $\rho_{\theta}(s \circ \alpha) = \rho_{\theta}(T) = (S_{\theta}, \emptyset)$. On the other hand $\phi_{\theta}(s) \circ \rho_{\theta}(\alpha) = \psi_{\theta}^{s} \circ (A_{\theta}^{\alpha}, B_{\theta}^{\alpha}) = (C, D)$ where $C = \{ \overline{q} \in S_{\theta} : \psi_{\theta}^{s}(\overline{q}) \in A_{\theta}^{\alpha} \}$ and $D = \{ \overline{q} \in S_{\theta} : \psi_{\theta}^{s}(\overline{q}) \in B_{\theta}^{\alpha} \}$. We have to show that $$(C, D) = (S_{\theta}, \emptyset).$$
 We first show by contradiction that $D = \emptyset$. If $D \neq \emptyset$ consider $\overline{q} \in D$
       \[\begin{array}{ccll}
         &\Rightarrow& \psi_{\theta}^{s}(\overline{q}) \in B_{\theta}^{\alpha} \\
         &\Rightarrow& \overline{q \cdot s} \in B_{\theta}^{\alpha} \\
         &\Rightarrow& ((q \cdot s) \circ \alpha, F) \in \theta \\
         &\Rightarrow& (q \circ (s \circ \alpha), F) \in \theta & \text{(using \eqref{EM2})} \\
         &\Rightarrow& (q \circ T, F) \in \theta
       \end{array}\]
 which is a contradiction to \pref{Prop.Rho.Hom}(ii). \\
 In order to show that $C = S_{\theta}$ consider $\overline{q} \in S_{\theta}$ that is $(q, \bot) \notin E_{\theta}$
       \[\begin{array}{ccll}
        &\Rightarrow& (q \circ T, T) \in \theta & \text{(using \pref{Prop.Rho.Hom}(iv))} \\
        &\Rightarrow& (q \circ (s \circ \alpha), T) \in \theta \\
        &\Rightarrow& ((q \cdot s) \circ \alpha, T) \in \theta & \text{(using \eqref{EM2})} \\
        &\Rightarrow& \overline{q \cdot s} \in A_{\theta}^{\alpha} \\
        &\Rightarrow& \psi_{\theta}^{s}(\overline{q}) \in A_{\theta}^{\alpha} \\
        &\Rightarrow& \overline{q} \in C.
      \end{array}\]

 \noindent \emph{Subcase 2}: $s \circ \alpha = F$. Then $\rho_{\theta}(s \circ \alpha) = \rho_{\theta}(F) = (\emptyset, S_{\theta})$ while $\phi_{\theta}(s) \circ \rho_{\theta}(\alpha) = \psi_{\theta}^{s} \circ (A_{\theta}^{\alpha}, B_{\theta}^{\alpha}) = (C, D)$ where $C = \{ \overline{q} \in S_{\theta} : \psi_{\theta}^{s}(\overline{q}) \in A_{\theta}^{\alpha} \}$ and $D = \{ \overline{q} \in S_{\theta} : \psi_{\theta}^{s}(\overline{q}) \in B_{\theta}^{\alpha} \}$. We have to show that $$(C, D) = (\emptyset, S_{\theta}).$$
 We first show $C = \emptyset$ by contradiction. If $C \neq \emptyset$ consider $\overline{q} \in C$
       \[\begin{array}{ccll}
        &\Rightarrow& \psi_{\theta}^{s}(\overline{q}) \in A_{\theta}^{\alpha} \\
        &\Rightarrow& \overline{q \cdot s} \in A_{\theta}^{\alpha} \\
        &\Rightarrow& ((q \cdot s) \circ \alpha, T) \in \theta \\
        &\Rightarrow& (q \circ (s \circ \alpha), T) \in \theta & \text{(using \eqref{EM2})} \\
        &\Rightarrow& (q \circ F, T) \in \theta \\
        &\Rightarrow& (q \circ T, F) \in \theta
      \end{array}\]
 which is a contradiction to \pref{Prop.Rho.Hom}(ii). \\
 In order to show that $D = S_{\theta}$ consider $\overline{q} \in S_{\theta}$ that is $(q, \bot) \notin E_{\theta}$.
       \[\begin{array}{ccll}
        &\Rightarrow& (q \circ F, F) \in \theta & \text{(using \pref{Prop.Rho.Hom}(iv))} \\
        &\Rightarrow& (q \circ (s \circ \alpha), F) \in \theta \\
        &\Rightarrow& ((q \cdot s) \circ \alpha, F) \in \theta & \text{(using \eqref{EM2})} \\
        &\Rightarrow& \overline{q \cdot s} \in B_{\theta}^{\alpha} \\
        &\Rightarrow& \psi_{\theta}^{s}(\overline{q}) \in B_{\theta}^{\alpha} \\
        &\Rightarrow& \overline{q} \in D
      \end{array}\]
 which completes the proof for the case where $\alpha \notin \{ T, F \}$ and $s \circ \alpha \in \{ T, F \}$.\\

 \noindent \emph{Case IV}: $\alpha \in \{ T, F \}, s \circ \alpha \in \{ T, F \}$. Note that $s \circ T \neq F$ as a consequence of  \pref{Prop.Rho.Hom}(ii). If $s \circ T = F$ then as $\theta$ is a congruence, $(F, F) \in \theta \Rightarrow (s \circ T, F) \in \theta$, a contradiction to \pref{Prop.Rho.Hom}(ii). Similarly we have $s \circ F \neq T$. In view of the above it suffices to consider the following cases: \\

 \noindent \emph{Subcase 1}: $\alpha = T, s \circ \alpha = T$. Then $\rho_{\theta}(s \circ \alpha) = \rho_{\theta}(T) = (S_{\theta}, \emptyset)$ and $\phi_{\theta}(s) \circ \rho_{\theta}(\alpha) = \psi_{\theta}^{s} \circ (S_{\theta}, \emptyset) = (C, D)$ where $C = \{ \overline{q} \in S_{\theta} : \psi_{\theta}^{s}(\overline{q}) \in S_{\theta} \}$ and $D = \emptyset$. Thus it suffices to show that $C = S_{\theta}$. Let $\overline{q} \in S_{\theta}$ that is $(q, \bot) \notin E_{\theta}$
       \[\begin{array}{ccll}
        &\Rightarrow& (q \circ T, T) \in \theta & \text{(using \pref{Prop.Rho.Hom}(iv))} \\
        &\Rightarrow& (q \circ (s \circ T), T) \in \theta \\
        &\Rightarrow& ((q \cdot s) \circ T, T) \in \theta & \text{(using \eqref{EM2})} \\
        &\Rightarrow& (q \cdot s, \bot) \notin E_{\theta} & \text{(using \pref{Prop.Rho.Hom}(iv))} \\
        &\Rightarrow& \overline{q \cdot s} \in S_{\theta} \\
        &\Rightarrow& \psi_{\theta}^{s}(\overline{q}) \in S_{\theta} \\
        &\Rightarrow& \overline{q} \in C.
      \end{array}\]
 Thus $C = S_{\theta}$. \\

 \noindent \emph{Subcase 2}: $\alpha = F, s \circ \alpha = F$. Then $\rho_{\theta}(s \circ \alpha) = \rho_{\theta}(F) = (\emptyset, S_{\theta})$ and $\phi_{\theta}(s) \circ \rho_{\theta}(\alpha) = \psi_{\theta}^{s} \circ (\emptyset, S_{\theta}) = (C, D)$ where $C = \emptyset$ and $D = \{ \overline{q} \in S_{\theta} : \psi_{\theta}^{s}(\overline{q}) \in S_{\theta} \}$. The proof follows along similar lines as above. In order to show that $D = S_{\theta}$ consider $\overline{q} \in S_{\theta}$ that is $(q, \bot) \notin E_{\theta}$
       \[\begin{array}{ccll}
        &\Rightarrow& (q \circ F, F) \in \theta & \text{(using \pref{Prop.Rho.Hom}(iv))} \\
        &\Rightarrow& (q \circ (s \circ F), F) \in \theta \\
        &\Rightarrow& ((q \cdot s) \circ F, F) \in \theta & \text{(using \eqref{EM2})} \\
        &\Rightarrow& (q \cdot s, \bot) \notin E_{\theta} & \text{(using \pref{Prop.Rho.Hom}(iv))} \\
        &\Rightarrow& \overline{q \cdot s} \in S_{\theta} \\
        &\Rightarrow& \psi_{\theta}^{s}(\overline{q}) \in S_{\theta} \\
        &\Rightarrow& \overline{q} \in D.
      \end{array}\]
 Hence $D = S_{\theta}$ which completes the proof.

 Thus $(\phi_{\theta}, \rho_{\theta})$ is a homomorphism of $C$-monoids.
\end{proof}

\begin{proposition} \label{PropAlphaTTheta}
 For all $\alpha \in M$ the following statements hold:
 \begin{enumerate}[\rm(i)]
  \item $\rho_{\theta}(\alpha) = (S_{\theta}, \emptyset) \Rightarrow (\alpha, T) \in \theta$.
  \item $\rho_{\theta}(\alpha) = (\emptyset, S_{\theta}) \Rightarrow (\alpha, F) \in \theta$.
 \end{enumerate}
\end{proposition}

\begin{proof}$\;$
 \begin{enumerate}[\rm(i)]
  \item If $\alpha = T$ then the result is obvious. Suppose that $\alpha \neq T$ and $\rho_{\theta}(\alpha) = (A_{\theta}^{\alpha}, B_{\theta}^{\alpha}) = (S_{\theta}, \emptyset)$. It follows that $(t \circ \alpha, T) \in \theta$ for all $\overline{t} \in S_{\theta}$. Using \pref{Prop.Rho.Hom}(vi) and \eqref{EM1} we have $\overline{1} \in S_{\theta}$ and so $(1 \circ \alpha, T) = (\alpha, T) \in \theta$.
  \item Along similar lines if $\alpha \neq F$ then $\rho_{\theta}(\alpha) = (A_{\theta}^{\alpha}, B_{\theta}^{\alpha}) = (\emptyset, S_{\theta})$ gives $(t \circ \alpha, F) \in \theta$ for all $\overline{t} \in S_{\theta}$. Using \pref{Prop.Rho.Hom}(vi) and \eqref{EM1} we have $\overline{1} \in S_{\theta}$ and so $(1 \circ \alpha, F) = (\alpha, F) \in \theta$.
 \end{enumerate}
\end{proof}

\begin{lemma} \label{LemmaPhiTheta}
 For every $s, t \in S_{\bot}$ where $s \neq t$ there exists a maximal congruence $\theta$ on $M$ such that $\phi_{\theta}(s) \neq \phi_{\theta}(t)$.
\end{lemma}

\begin{proof}
 Using \pref{PropCollection}(iv) we have $\bigcap E_{\theta} = \Delta_{S_{\bot}}$ and so since $s \neq t$ there exists a maximal congruence $\theta$ on $M$ such that $(s, t) \notin E_{\theta}$, i.e., $\overline{s} \neq \overline{t}$. For this $\theta$, consider $\phi_{\theta} : S_{\bot} \rightarrow \mathcal{T}_{o}(S_{{\theta_{\bot}}})$. Then $\phi_{\theta}(s) = \psi_{\theta}^{s}$, $\phi_{\theta}(t) = \psi_{\theta}^{t}$. For $\overline{1} \in S_{{\theta}_{\bot}}$ we have $\psi_{\theta}^{s}(\overline{1}) = \overline{1 \cdot s} = \overline{s}$ while $\psi_{\theta}^{t}(\overline{1}) = \overline{1 \cdot t} = \overline{t}$. Since $\overline{s} \neq \overline{t}$ it follows that $\phi_{\theta}(s) \neq \phi_{\theta}(t)$.
\end{proof}

\begin{lemma} \label{LemmaRhoTheta}
 For every $\alpha, \beta \in M$ where $\alpha \neq \beta$ there exists a maximal congruence $\theta$ on $M$ such that $\rho_{\theta}(\alpha) \neq \rho_{\theta}(\beta)$.
\end{lemma}

\begin{proof}
  Using \pref{PropCollection}(v) since $\alpha \neq \beta$ there exists a maximal congruence $\theta$ on $M$ such that $(\alpha, \beta) \notin \theta$. We  show that $\rho_{\theta}(\alpha) \neq \rho_{\theta}(\beta)$. If $\alpha$ or $\beta$ is in $\{ T, F \}$ but $\rho_{\theta}(\alpha) = \rho_{\theta}(\beta)$ then using \pref{PropAlphaTTheta} we have $(\alpha, \beta) \in \theta$, a contradiction. In the case where $\alpha, \beta \notin \{ T, F \}$ we show that $$(A_{\theta}^{\alpha}, B_{\theta}^{\alpha}) \neq (A_{\theta}^{\beta}, B_{\theta}^{\beta})$$ by showing that either $A_{\theta}^{\alpha} \neq A_{\theta}^{\beta}$ or that $B_{\theta}^{\alpha} \neq B_{\theta}^{\beta}$. Owing to \pref{PropMaxTheta} it suffices to consider the following three cases:

 \noindent \emph{Case I}: $(\alpha, T) \in \theta$. Note that \pref{Prop.Rho.Hom}(vi) gives $\overline{1} \in S_{\theta}$. Thus we have $\overline{1} \in S_{\theta}$ for which $1 \circ \alpha = \alpha \in \overline{T}^{\theta}$ and so $\overline{1} \in A_{\theta}^{\alpha}$. However $\overline{1} \notin A_{\theta}^{\beta}$ since $(\alpha, \beta) \notin \theta$. \\

 \noindent \emph{Case II}: $(\alpha, F) \in \theta$. Along similar lines for $\overline{1} \in S_{\theta}$ we have $1 \circ \alpha = \alpha \in \overline{F}^{\theta}$ and so $\overline{1} \in B_{\theta}^{\alpha}$. It is clear that $\overline{1} \notin B_{\theta}^{\beta}$ since $(\alpha, \beta) \notin \theta$. \\

  \noindent \emph{Case III}: $(\alpha, U) \in \theta$. In view of \pref{PropMaxTheta} it suffices to consider the following cases:

   \noindent \emph{Subcase 1}: $(\beta, T) \in \theta$. As earlier we have $\overline{1} \in A_{\theta}^{\beta} \setminus A_{\theta}^{\alpha}$. \\

   \noindent \emph{Subcase 2}: $(\beta, F) \in \theta$. It is clear that $\overline{1} \in B_{\theta}^{\beta} \setminus B_{\theta}^{\alpha}$.

 Thus $\rho_{\theta}(\alpha) \neq \rho_{\theta}(\beta)$ which completes the proof.
\end{proof}

\subsection{Embedding into a functional $C$-monoid}
\label{embed_func}
Let $\{\theta\}$ be the collection of all maximal congruences of $M$. Define the set $X$ to be the disjoint union of $S_{\theta}$ taken over all maximal congruences of $M$, written
\begin{equation}\label{NotationDisjoint}
 X = \bigsqcup_{\theta} S_{\theta}
\end{equation}
Set $X_{\bot} = X \cup \{ \bot \}$ with base point $\bot \notin X$. For notational convenience we use the same symbol $\bot$ in $X_{\bot}$ as well as in $S_{\bot}$. Which $\bot$ we are referring to will be clear from the context of the statement.

In this subsection we obtain monomorphisms $\phi : S_{\bot} \rightarrow \mathcal{T}_{o}(X_{\bot})$ and $\rho : M \rightarrow \mathbb{3}^{X}$, using which we establish that $(S_{\bot}, M)$ can be embedded into the functional $C$-monoid $\big( \mathcal{T}_{o}(X_{\bot}), \mathbb{3}^{X} \big)$.

\begin{remark}$\;$ \label{RemarkDisjoint}
 \begin{enumerate}[\rm(i)]
  \item Let $q \in S$ be fixed. For different $\theta$'s the representation of classes $\overline{q}^{_{E_{\theta}}}$'s are different in the disjoint union $X$ of $S_{\theta}$'s.
  \item Let $\{A_{\lambda}\}, \{B_{\lambda}\}$ be two families of sets indexed over $\Lambda$. Then $\displaystyle\bigsqcup_\lambda (A_{\lambda} \cap B_{\lambda}) = \Big(\bigsqcup_\lambda A_{\lambda}\Big) \cap \Big(\bigsqcup_\lambda B_{\lambda}\Big)$ and $\displaystyle\bigsqcup_\lambda (A_{\lambda} \cup B_{\lambda}) = \Big(\bigsqcup_\lambda A_{\lambda}\Big) \cup \Big(\bigsqcup_\lambda B_{\lambda}\Big)$.
 \end{enumerate}
\end{remark}

\begin{notation}$\;$
 \begin{enumerate}[\rm(i)]
  \item For the pair of sets $(A, B)$, we denote by $\pi_{1}(A, B)$ the first component $A$, and by $\pi_{2}(A, B)$ the second component $B$.
  \item For a family of pairs of sets $(A_{\lambda}, B_{\lambda})$ where $\lambda \in \Lambda$ we denote by $\displaystyle\bigsqcup_\lambda (A_{\lambda}, B_{\lambda})$ the pair of sets $\displaystyle\Big(\bigsqcup_\lambda A_{\lambda}, \bigsqcup_\lambda B_{\lambda}\Big)$.
 \end{enumerate}
\end{notation}

\begin{lemma} \label{LemmaPhiMono}
 Consider $\phi : S_{\bot} \rightarrow \mathcal{T}_{o}(X_{\bot})$ given by
 \begin{equation*}
  (\phi(s))(x) = \begin{cases}
                    (\phi_{\theta}(s))(\overline{q}^{_{E_{\theta}}}), & \text{ if } x = \overline{q}^{_{E_{\theta}}} \in S_{\theta} \text{ and } (\phi_{\theta}(s))(\overline{q}^{_{E_{\theta}}}) \neq \overline{\bot}^{_{E_{\theta}}}; \\
                    \bot, & \text{ otherwise.}
                   \end{cases}
 \end{equation*}
 Then $\phi$ is a monoid monomorphism that maps the zero (and base point) of $S_{\bot}$ to that of $\mathcal{T}_{o}(X_{\bot})$, that is $\bot \mapsto \zeta_{\bot}$.
\end{lemma}

\begin{proof}
  It is clear that $\phi$ is well-defined and that $\phi(s) \in \mathcal{T}_{o}(X_{\bot})$ since $(\phi(s))(\bot) = \bot$. \\

  \emph{Claim: $\phi$ is injective.} Let $s \neq t \in S_{\bot}$. Using \lref{LemmaPhiTheta} there exists a maximal congruence $\theta$ on $M$ such that $\phi_{\theta}(s) \neq \phi_{\theta}(t)$. Hence there exists a $\overline{q}^{_{E_{\theta}}} (\neq \overline{\bot}^{_{E_{\theta}}})$ such that $(\phi_{\theta}(s))(\overline{q}) \neq (\phi_{\theta}(t))(\overline{q})$. By extrapolation it follows that $(\phi(s))(\overline{q}) \neq (\phi(t))(\overline{q})$ and so $\phi(s) \neq \phi(t)$. \\

  \emph{Claim: $\phi(\bot) = \zeta_{\bot}$.} Using \pref{Prop.phitheta} we have $\phi_{\theta}(\bot) = \zeta_{\overline{\bot}^{_{E_{\theta}}}}$ for all $\theta$ and so by definition $(\phi(\bot))(x) = \bot$ for all $x \in X_{\bot}$. \\

  \emph{Claim: $\phi(1) = id_{X_{\bot}}$.} It is clear that $(\phi(1))(\bot) = \bot$. Consider $\overline{q} \in X$ that is $\overline{q}^{_{E_{\theta}}} \in S_{\theta}$ for some $\theta$. Then by \pref{Prop.phitheta} we have $(\phi(1))(\overline{q}^{_{E_{\theta}}}) = (\phi_{\theta}(1))(\overline{q}^{_{E_{\theta}}}) = \overline{q}^{_{E_{\theta}}}$ and hence $\phi(1) = id_{X_{\bot}}$. \\

  \emph{Claim: $\phi(s \cdot t) = \phi(s) \cdot \phi(t)$.} Clearly $(\phi(s \cdot t))(\bot) = \bot = (\phi(s) \cdot \phi(t))(\bot)$. Let $\overline{q} \in X$ that is $\overline{q}^{_{E_{\theta}}} \in S_{\theta}$ for some $\theta$. Suppose that $(\phi(s \cdot t))(\overline{q}) = \bot$ so that $(\phi_{\theta}(s \cdot t))(\overline{q}) = \overline{\bot}$
       \[\begin{array}{ccll}
         &\Rightarrow& ((\phi_{\theta}(s) \cdot \phi_{\theta}(t))(\overline{q}) = \overline{\bot} & \text{(using \pref{Prop.phitheta})} \\
         &\Rightarrow& \phi_{\theta}(t)(\phi_{\theta}(s)(\overline{q})) = \overline{\bot} \\
         &\Rightarrow& \phi(t)(\phi_{\theta}(s)(\overline{q})) = \bot
       \end{array}\]
  Noting that there are only two possibilities for $\phi(s)(\overline{q})$ we see that if $\phi(s)(\overline{q}) = \phi_{\theta}(s)(\overline{q})$ then we are through. On the other hand if $\phi(s)(\overline{q}) = \bot$ that is $\phi_{\theta}(s)(\overline{q}) = \overline{\bot}$ then we have $(\phi(s \cdot t))(\overline{q}) = \bot = (\phi(s) \cdot \phi(t))(\overline{q})$ which completes the proof in this case.

  Consider the case where $(\phi(s \cdot t))(\overline{q}) \neq \bot$. Using \pref{Prop.phitheta} it follows that $(\phi(s \cdot t))(\overline{q}) = (\phi_{\theta}(s \cdot t))(\overline{q}) = (\phi_{\theta}(s) \cdot \phi_{\theta}(t))(\overline{q}) = \phi_{\theta}(t)(\phi_{\theta}(s)(\overline{q}))$ and so $(\phi_{\theta}(s))(\overline{q}) \neq \overline{\bot}$. Consequently $\phi(t)(\phi(s)(\overline{q})) = \phi_{\theta}(t)(\phi_{\theta}(s)(\overline{q}))$ since $(\phi_{\theta}(s))(\overline{q}) \neq \overline{\bot}$. It follows that $(\phi(s \cdot t))(\overline{q}) = (\phi(s) \cdot \phi(t))(\overline{q})$ which completes the proof.
\end{proof}

\begin{lemma} \label{LemmaRhoMono}
 The function $\rho : M \rightarrow \mathbb{3}^{X}$ defined by
 \begin{equation*}
  \rho(\alpha) = \displaystyle \sqcup_{\theta} \rho_{\theta}(\alpha)
 \end{equation*}
 is a monomorphism of $C$-algebras with $T, F, U$.
\end{lemma}

\begin{proof}$\;$
  \emph{Claim: $\rho$ is well defined.} Let $\alpha \in M$. Using \rref{RemarkDisjoint}(i) we have $\pi_{1}(\rho(\alpha)) \cap \pi_{2}(\rho(\alpha)) = \emptyset$ due to the distinct representation of equivalence classes. Also by \pref{Prop.rhotheta} we have $\pi_{1}(\rho_{\theta}(\alpha)), \pi_{2}(\rho_{\theta}(\alpha)) \subseteq S_{\theta}$ and so $\bot \notin \pi_{1}(\rho(\alpha)) \cup \pi_{2}(\rho(\alpha))$ that is $\rho(\alpha)$ is can be identified with a pair of sets over $X$. \\

  \emph{Claim: $\rho$ is injective.} Let $\alpha \neq \beta \in M$. By \lref{LemmaRhoTheta} there exists a $\theta$ such that $\rho_{\theta}(\alpha) \neq \rho_{\theta}(\beta)$. Without loss of generality we infer that there exists a $\overline{q}^{_{E_{\theta}}} \in \pi_{1}(\rho_{\theta}(\alpha)) \setminus \pi_{1}(\rho_{\theta}(\beta))$. Since $\rho(\alpha)$ is formed by taking the disjoint union of the individual images under $\rho_{\theta}(\alpha)$, using \rref{RemarkDisjoint}(i) we can say that $\overline{q} \in \pi_{1}(\rho(\alpha)) \setminus \pi_{1}(\rho(\beta))$ that is $\rho(\alpha) \neq \rho(\beta)$. \\

  \emph{Claim: $\rho$ preserves the constants $T, F, U$.}   It follows easily from \pref{Prop.rhotheta} that $\rho(T) = (X, \emptyset)$, $\rho(F) = (\emptyset, X)$ and $\rho(U) = (\emptyset, \emptyset)$. \\

  \emph{Claim: $\rho(\neg \alpha) = \neg (\rho(\alpha))$.} If $\alpha \in \{ T, F \}$ then the result is obvious. If $\alpha \notin \{ T, F \}$ then $\neg \alpha \notin \{ T, F \}$. Using \pref{Prop.rhotheta} we have $\rho(\neg \alpha) = (\sqcup A_{\theta}^{\neg \alpha}, \sqcup B_{\theta}^{\neg \alpha}) = (\sqcup B_{\theta}^{\alpha}, \sqcup A_{\theta}^{\alpha})$. Thus $\rho(\neg \alpha) = (\sqcup B_{\theta}^{\alpha}, \sqcup A_{\theta}^{\alpha}) = \neg (\rho(\alpha))$. \\

  \emph{Claim: $\rho(\alpha \wedge \beta) = \rho(\alpha) \wedge \rho(\beta)$.} In view of \rref{RemarkDisjoint}(ii) we have $\sqcup ((A_{\lambda}, B_{\lambda}) \wedge (C_{\lambda}, D_{\lambda})) = (\sqcup A_{\gamma}, \sqcup B_{\gamma}) \wedge (\sqcup C_{\gamma}, \sqcup D_{\gamma})$ for the family of pairs of sets $(A_{\lambda}, B_{\lambda}), (C_{\lambda}, D_{\lambda})$ where $\lambda \in \Lambda$ over $X$. In view of the above and \pref{Prop.rhotheta} we have $\sqcup \rho_{\theta}(\alpha \wedge \beta) = \sqcup (\rho_{\theta}(\alpha) \wedge \rho_{\theta}(\beta)) = (\sqcup \rho_{\theta}(\alpha)) \wedge (\sqcup \rho_{\theta}(\beta)) = \rho(\alpha) \wedge \rho(\beta)$ which completes the proof.
\end{proof}

\begin{lemma} \label{LemmaPhiRho}
 The pair $(\phi, \rho)$ is a $C$-monoid monomorphism from $(S_{\bot}, M)$ to the functional $C$-monoid $\big( \mathcal{T}_{o}(X_{\bot}), \mathbb{3}^{X} \big)$.
\end{lemma}

\begin{proof}
 In view of \lref{LemmaPhiMono} and \lref{LemmaRhoMono} it suffices to show $\phi(\alpha[s, t]) = (\rho(\alpha))[\phi(s), \phi(t)]$ and $\rho(s \circ \alpha) = \phi(s) \circ \rho(\alpha)$. \\

 In order to show that $\phi(\alpha[s, t]) = (\rho(\alpha))[\phi(s), \phi(t)]$ we show that $\phi(\alpha[s, t])(x) = (\rho(\alpha))[\phi(s), \phi(t)](x)$ for all $x \in X_{\bot}$. Thus we have the following cases: \\

 \noindent \emph{Case I}: $x = \bot$. It is clear that $\phi(\alpha[s, t])(\bot) = \bot = (\rho(\alpha))[\phi(s), \phi(t)](\bot)$ since $\pi_{1}(\rho(\alpha)), \pi_{2}(\rho(\alpha)) \subseteq X$ and $\bot \notin X$. \\

 \noindent \emph{Case II}: $x \in X$. Consider $\overline{q} \in X$ that is $\overline{q}^{_{E_{\theta}}} \in S_{\theta}$ for some $\theta$. We have the following subcases: \\

 \noindent \emph{Subcase 1}: $\phi(\alpha[s, t])(\overline{q}) = \bot$. then $\phi_{\theta}(\alpha[s, t])(\overline{q}) = \overline{\bot}$ and so using \lref{LemmaPhiRhoTheta} we have $\phi_{\theta}(\alpha[s, t])(\overline{q}) = \overline{\bot} = (\rho_{\theta}(\alpha))[\phi_{\theta}(s), \phi_{\theta}(t)](\overline{q})$. It follows that either $\overline{q} \notin \pi_{1}(\rho_{\theta}(\alpha)) \cup \pi_{2}(\rho_{\theta}(\alpha))$ or that $\overline{q} \in \pi_{1}(\rho_{\theta}(\alpha))$ and $\phi_{\theta}(s)(\overline{q}) = \overline{\bot}$ or, similarly, that $\overline{q} \in \pi_{2}(\rho_{\theta}(\alpha))$ and $\phi_{\theta}(t)(\overline{q}) = \overline{\bot}$. Thus we have the following:
 \begin{description}
  \item[$\overline{q} \notin \pi_{1}(\rho_{\theta}(\alpha)) \cup \pi_{2}(\rho_{\theta}(\alpha))$] In view of \rref{RemarkDisjoint}(i) it follows that $\overline{q} \notin \pi_{1}(\rho(\alpha)) \cup \pi_{2}(\rho(\alpha))$ and so $(\rho(\alpha))[\phi(s), \phi(t)](\overline{q}) = \bot$.
  \item[$\overline{q} \in \pi_{1}(\rho_{\theta}(\alpha))$ and $\phi_{\theta}(s)(\overline{q}) = \overline{\bot}$] Then $\overline{q} \in \pi_{1}(\rho(\alpha))$ and $\phi(s)(\overline{q}) = \bot$ and so $(\rho(\alpha))[\phi(s), \phi(t)](\overline{q}) = \bot$.
  \item[$\overline{q} \in \pi_{2}(\rho_{\theta}(\alpha))$ and $\phi_{\theta}(t)(\overline{q}) = \overline{\bot}$] Along similar lines we have $(\rho(\alpha))[\phi(s), \phi(t)](\overline{q}) = \bot$. \\
 \end{description}

 \noindent \emph{Subcase 2}: $\phi(\alpha[s, t])(\overline{q}) \neq \bot$. Then $\phi(\alpha[s, t])(\overline{q}) = \phi_{\theta}(\alpha[s, t])(\overline{q})$ and so using \lref{LemmaPhiRhoTheta} we have $\phi(\alpha[s, t])(\overline{q}) = (\rho_{\theta}(\alpha))[\phi_{\theta}(s), \phi_{\theta}(t)](\overline{q})$. It follows that
 \begin{equation*}
  \phi(\alpha[s, t])(\overline{q}) = (\rho_{\theta}(\alpha))[\phi_{\theta}(s), \phi_{\theta}(t)](\overline{q}) = \begin{cases}
                                                                                                                  \phi_{\theta}(s)(\overline{q}), & \text{ if } \overline{q} \in \pi_{1}(\rho_{\theta}(\alpha)); \\
                                                                                                                  \phi_{\theta}(t)(\overline{q}), & \text{ if } \overline{q} \in \pi_{2}(\rho_{\theta}(\alpha)); \\
                                                                                                                  \bot, & \text{ otherwise.}
                                                                                                                  \end{cases}
 \end{equation*}
 \begin{description}
  \item[$\overline{q} \in \pi_{1}(\rho_{\theta}(\alpha))$] It follows that $\overline{q} \in \pi_{1}(\rho(\alpha))$ and so $(\rho(\alpha))[\phi(s), \phi(t)](\overline{q}) = \phi(s)(\overline{q})$. Note that $\phi_{\theta}(s)(\overline{q}) \neq \overline{\bot}$ else $\phi(\alpha[s, t])(\overline{q}) = \bot$, a contradiction. Thus $\phi(s)(\overline{q}) = \phi_{\theta}(s)(\overline{q})$ so that $\phi(\alpha[s, t])(\overline{q}) = (\rho(\alpha))[\phi(s), \phi(t)](\overline{q})$.
  \item[$\overline{q} \in \pi_{2}(\rho_{\theta}(\alpha))$] The proof follows along similar lines as above.
  \item[$\overline{q} \notin (\pi_{1}(\rho_{\theta}(\alpha)) \cup \pi_{2}(\rho_{\theta}(\alpha)))$] This case cannot occur since we assumed that $\phi(\alpha[s, t])(\overline{q}) \neq \bot$.
 \end{description}
 Thus $\phi(\alpha[s, t]) = (\rho(\alpha))[\phi(s), \phi(t)]$. \\

 We now show that $\rho(s \circ \alpha) = \phi(s) \circ \rho(\alpha)$. In order to prove this we proceed by showing that $$\pi_{i}(\rho(s \circ \alpha)) = \pi_{i}(\phi(s) \circ \rho(\alpha))$$ for $i \in \{ 1, 2 \}$.

 Let $\overline{q} \in \pi_{1}(\rho(s \circ \alpha)) = \sqcup \pi_{1}(\rho_{\theta}(s \circ \alpha))$. Then $\overline{q}^{_{E_{\theta}}} \in S_{\theta}$ for some $\theta$ and $\overline{q}^{_{E_{\theta}}} \in \pi_{1}(\rho_{\theta}(s \circ \alpha))$
       \[\begin{array}{ccll}
         &\Rightarrow& \overline{q}^{_{E_{\theta}}} \in \pi_{1}(\phi_{\theta}(s) \circ \rho_{\theta}(\alpha)) & \text{(using \lref{LemmaPhiRhoTheta})} \\
         &\Rightarrow& \phi_{\theta}(s)(\overline{q}^{_{E_{\theta}}}) \in \pi_{1}(\rho_{\theta}(\alpha)) \subseteq S_{\theta} \\
         &\Rightarrow& \phi_{\theta}(s)(\overline{q}^{_{E_{\theta}}}) \neq \overline{\bot} \\
         &\Rightarrow& \phi(s)(\overline{q}^{_{E_{\theta}}}) = \phi_{\theta}(s)(\overline{q}^{_{E_{\theta}}}) \\
         &\Rightarrow& \phi(s)(\overline{q}^{_{E_{\theta}}}) \in \sqcup \pi_{1}(\rho_{\theta}(\alpha)) \\
         &\Rightarrow& \phi(s)(\overline{q}^{_{E_{\theta}}}) \in \pi_{1}(\rho(\alpha)) \\
         &\Rightarrow& \overline{q}^{_{E_{\theta}}} \in \pi_{1}(\phi(s) \circ \rho(\alpha))
       \end{array}\]
 and so $\pi_{1}(\rho(s \circ \alpha)) \subseteq \pi_{1}(\phi(s) \circ \rho(\alpha))$.

 For the reverse inclusion assume that $\overline{q} \in \pi_{1}(\phi(s) \circ \rho(\alpha))$. Consequently we have $\overline{q}^{_{E_{\theta}}} \in S_{\theta}$ for some $\theta$ and $\phi(s)(\overline{q}^{_{E_{\theta}}}) \in \pi_{1}(\rho(\alpha)) \subseteq X$
 \[\begin{array}{ccll}
         &\Rightarrow& \phi(s)(\overline{q}^{_{E_{\theta}}}) \neq \bot \\
         &\Rightarrow& \phi(s)(\overline{q}^{_{E_{\theta}}}) = \phi_{\theta}(s)(\overline{q}^{_{E_{\theta}}}) (\neq \overline{\bot}^{_{E_{\theta}}}) \\
         &\Rightarrow& \phi_{\theta}(s)(\overline{q}^{_{E_{\theta}}}) \in \pi_{1}(\rho_{\theta}(\alpha)) & \text{(using \rref{RemarkDisjoint}(i))} \\
         &\Rightarrow& \overline{q}^{_{E_{\theta}}} \in \pi_{1}(\phi_{\theta}(s) \circ \rho_{\theta}(\alpha)) \\
         &\Rightarrow& \overline{q}^{_{E_{\theta}}} \in \pi_{1}(\rho_{\theta}(s \circ \alpha)) & \text{(using \lref{LemmaPhiRho})} \\
         &\Rightarrow& \overline{q}^{_{E_{\theta}}} \in \sqcup \pi_{1}(\rho_{\theta}(s \circ \alpha)) = \pi_{1}(\rho(s \circ \alpha)) \\
       \end{array}\]
 from which it follows that $\pi_{1}(\phi(s) \circ \rho(\alpha)) \subseteq \pi_{1}(\rho(s \circ \alpha))$.
 Proceeding along exactly the same lines we can show that $\pi_{2}(\rho(s \circ \alpha)) = \pi_{2}(\phi(s) \circ \rho(\alpha))$ which completes the proof.
\end{proof}

\subsection{Proof of Theorem \ref{main_theorem}}
\label{proof_main_theorem}

 Let $\{\theta\}$ be the collection of all maximal congruences of $M$. Consider the set $X$ as in \eqref{NotationDisjoint}. The functions $\phi : S_{\bot} \rightarrow \mathcal{T}_{o}(X_{\bot})$ and $\rho : M \rightarrow \mathbb{3}^{X}$ as defined in \lref{LemmaPhiMono} and \lref{LemmaRhoMono}, respectively, are monomorphisms. Further, by \lref{LemmaPhiRho}, the pair $(\phi, \rho)$ is a monomorphism  from $(S_{\bot}, M)$ to the functional $C$-monoid $\big( \mathcal{T}_{o}(X_{\bot}), \mathbb{3}^{X} \big)$. From the construction of $X$ it is also evident  that if $M$ and $S_{\bot}$ are finite then there are only finitely many maximal congruences $\theta$ on $M$ and finitely many equivalence classes $E_{\theta}$ on $S_{\bot}$ and so $X$ must be finite.

\begin{corollary} \label{IdCmonoids}
 An identity is satisfied in every $C$-monoid $(S_{\bot}, M)$ where $M$ is an ada if and only if it is satisfied in all functional $C$-monoids.
\end{corollary}

In view of \cref{IdCmonoids} and \eqref{FunctionalCMonoidEq}, we have the following result.

\begin{corollary}
 In every $C$-monoid $(S_{\bot}, M)$ where $M$ is an ada we have \\ $(f \circ T)[f, f] = f$.
\end{corollary}

\section{Conclusion} \label{Conclusion}

The notion of $C$-sets axiomatize the program construct {\tt if-then-else} considered over possibly non-halting programs and non-halting tests. In this work, we extended the axiomatization to $C$-monoids which include the composition of programs as well as composition of programs with tests. For the class of $C$-monoids where the $C$-algebra is an ada we obtain a Cayley-type theorem which exhibits the embedding of such $C$-monoids into functional $C$-monoids. Using this, we obtain a mechanism to determine the equivalence of programs through functional $C$-monoids. It is desirable to achieve such a representation for the general class of $C$-monoids with no restriction on the $C$-algebra, which can be considered as future work. Note that the term $f \circ T$ in the standard functional model of a $C$-monoid represents the aspect of the domain of the function, as used in \cite{desharnais09,jackson15}. It is interesting to study the relation between these two concepts in the current set up.


\appendix

\section{Proofs} \label{SectionAnnexure}

\subsection{Verification of \eref{ExampleFunctionalCmonoid}} \label{VerFunctionalCmonoid}

\begin{flushleft}

We use the pairs of sets representation given by Guzm\'{a}n and Squier in \cite{guzman90} and identify $\alpha \in \mathbb{3}^{X}$ with a \emph{pair of sets} $(A, B)$ of $X$ where $A = \alpha^{-1}(T)$ and $B = \alpha^{-1}(F)$. In this representation ${\bf T} = (X, \emptyset), {\bf F} = (\emptyset, X)$ and ${\bf U} = (\emptyset, \emptyset)$. Thus the operation $\circ$ is given as follows:
 \begin{equation*}
  (f \circ \alpha)(x) =
  \begin{cases}
   T, & \text{ if } f(x) \in A; \\
   F, & \text{ if } f(x) \in B; \\
   U, & \text{ otherwise.}
  \end{cases}
 \end{equation*}
In other words $f \circ \alpha$ can be identified with the pair of sets $(C, D)$ where $C = \{ x \in X : f(x) \in A \}$ and $D = \{ x \in X : f(x) \in B \}$. \newline

Axiom \eqref{EM7}: Let $\alpha$ be identified with the pair of sets $(A, B)$. Then $\zeta_{\bot} \circ \alpha = (\emptyset, \emptyset) = {\bf U}$ as $\zeta_{\bot}(x) = \bot \notin (A \cup B)$. \newline

Axiom \eqref{EM8}: Consider ${\bf U} = (\emptyset, \emptyset)$. Then $f \circ {\bf U} = (\emptyset, \emptyset) = {\bf U}$. \newline

Axiom \eqref{EM1}:
  \begin{align*}
   (1 \circ \alpha)(x) & =
  \begin{cases}
   T, & \text{ if } id_{X_{\bot}}(x) \in A; \\
   F, & \text{ if } id_{X_{\bot}}(x) \in B; \\
   U, & \text{ otherwise}
  \end{cases} \\
   & = \alpha(x).
   \end{align*}
  Thus $1 \circ \alpha = \alpha$. \newline

Axiom \eqref{EM4}: Let $\alpha$ be identified with the pair of sets $(A, B)$. Then $f \circ \alpha = (C, D)$ where $C = \{ x \in X : f(x) \in A \}$ and $D = \{ x \in X : f(x) \in B \}$. Thus $\neg (f \circ \alpha) = (D, C)$. Also $f \circ (\neg \alpha) = f \circ (B, A) = (E, F)$ where $E = \{ x \in X : f(x) \in B \}$ and $F = \{ x \in X : f(x) \in A \}$. It follows that $(E, F) = (D, C)$. \newline

Axiom \eqref{EM3}: Let $\alpha, \beta$ be represented by the pairs of sets $(A_{1}, A_{2})$ and $(B_{1}, B_{2})$ respectively. Then $\alpha \wedge \beta = (A_{1} \cap B_{1}, A_{2} \cup (A_{1} \cap B_{2}))$. Also let $f \circ \alpha = (C_{1}, C_{2})$ where $C_{1} = \{ x \in X : f(x) \in A_{1} \}$ and $C_{2} = \{ x \in X : f(x) \in A_{2} \}$, and $f \circ \beta = (D_{1}, D_{2})$ where $D_{1} = \{ x \in X : f(x) \in B_{1} \}$ and $D_{2} = \{ x \in X : f(x) \in B_{2} \}$. Then $(C_{1}, C_{2}) \wedge (D_{1}, D_{2}) = (C_{1} \cap D_{1}, C_{2} \cup (C_{1} \cap D_{2}))$. Thus $C_{1} \cap D_{1} = \{ x \in X : f(x) \in A_{1} \cap B_{1} \}$ and $C_{2} \cup (C_{1} \cap D_{2}) = \{ x \in X : f(x) \in A_{2} \cup (A_{1} \cap B_{2}) \}$. Hence $f \circ (\alpha \wedge \beta) = (f \circ \alpha) \wedge (f \circ \beta)$. \newline

Axiom \eqref{EM2}: Consider $f, g \in \mathcal{T}_{o}(X_{\bot})$ and $\alpha \in \mathbb{3}^{X}$ represented by the pair of sets $(A, B)$.
  \begin{equation*}
   ((f \cdot g) \circ \alpha)(x) =
   \begin{cases}
    T, &\text{ if } g(f(x)) \in A; \\
    F, & \text{ if } g(f(x)) \in B; \\
    U, & \text{ otherwise.}
   \end{cases}
  \end{equation*}
  Let $g \circ \alpha = (C, D)$ where $C = \{ x \in X : g(x) \in A \}$ and $D = \{ x \in X : g(x) \in B \}$.
  \begin{equation*}
   (f \circ (g \circ \alpha))(x) =
   \begin{cases}
    T, & \text{ if } f(x) \in C; \\
    F, & \text{ if } f(x) \in D; \\
    U, & \text{ otherwise.}
   \end{cases}
  \end{equation*}
  We may consider the following three cases. \\
   \emph{Case I}: \emph{$x \in X$ such that $g(f(x)) \in A$:} Then $((f \cdot g) \circ \alpha)(x) = T$. Also $f(x) \in C$ as $g(f(x)) \in A$. Thus $(f \circ (g \circ \alpha))(x) = T$. \\
   \emph{Case II}: \emph{$x \in X$ such that $g(f(x)) \in B$:} Then $((f \cdot g) \circ \alpha)(x) = F$. Similarly $g(f(x)) \in B$ means that $f(x) \in D$. Thus $(f \circ (g \circ \alpha))(x) = F$. \\
   \emph{Case III}: \emph{$x \in X$ such that $g(f(x)) \notin (A \cup B)$:} Then $((f \cdot g) \circ \alpha)(x) = U$. Since $f(x)$ is in neither $C$ nor $D$ it follows that $(f \circ (g \circ \alpha))(x) = U$. \newline

Axiom \eqref{EM5}: Consider $\alpha \in \mathbb{3}^{X}$ represented by the pair of sets $(A, B)$.
  \begin{equation*}
   (\alpha[f, g] \cdot h)(x) = h(\alpha[f, g](x)) =
   \begin{cases}
    h(f(x)), & \text{ if } x \in A; \\
    h(g(x)), & \text{ if } x \in B; \\
    \bot, & \text{ otherwise.}
   \end{cases}
  \end{equation*}
  Hence $\alpha[f, g] \cdot h = \alpha[f \cdot h, g \cdot h]$. \newline

Axiom \eqref{EM9}: Let $\alpha \in \mathbb{3}^{X}$ be represented by the pair of sets $(A, B)$.
  \begin{equation*}
   (h \cdot \alpha[f, g])(x) = \alpha[f, g](h(x)) =
   \begin{cases}
    f(h(x)), & \text{ if } h(x) \in A; \\
    g(h(x)), & \text{ if } h(x) \in B; \\
    \bot, & \text{ otherwise.}
   \end{cases}
  \end{equation*}
  Let $h \circ \alpha$ be represented by the pair of sets $(C, D)$ where $C = \{ x \in X : h(x) \in A \}$ and $D = \{ x \in X : h(x) \in B \}$.
  \begin{align*}
   (h \circ \alpha)[h \cdot f, h \cdot g](x) & =
   \begin{cases}
    (h \cdot f)(x), & \text{ if } x \in C; \\
    (h \cdot g)(x), & \text{ if } x \in D; \\
    \bot, & \text{ otherwise} \\
   \end{cases} \\
    & =
    \begin{cases}
     f(h(x)), & \text{ if } h(x) \in A; \\
     g(h(x)), & \text{ if } h(x) \in B; \\
     \bot, & \text{ otherwise.}
    \end{cases}
  \end{align*}
  Thus $h \cdot \alpha[f, g] = (h \circ \alpha)[h \cdot f, h \cdot g]$. \newline

  Axiom \eqref{EM6}: Let $\alpha, \beta \in \mathbb{3}^{X}$ be represented by the pairs of sets $(A_{1}, A_{2})$ and $(B_{1}, B_{2})$ respectively. For $f, g \in \mathcal{T}_{o}(X_{\bot})$ we have the following:
  \begin{equation*}
   h(x) = \alpha[f, g](x) =
   \begin{cases}
    f(x), & \text{ if } x \in A_{1}; \\
    g(x), & \text{ if } x \in A_{2}; \\
    \bot, & \text{ otherwise.}
   \end{cases}
  \end{equation*}
  Also $h \circ \beta = (C_{1}, C_{2})$ where $C_{1} = \{ x \in X : h(x) \in B_{1} \}$ and $C_{2} = \{ x \in X : h(x) \in B_{2} \}$. Similarly $f \circ \beta = (D_{1}, D_{2})$ where $D_{1} = \{ x \in X : f(x) \in B_{1} \}$ and $D_{2} = \{ x \in X : f(x) \in B_{2} \}$. Let $g \circ \beta = (E_{1}, E_{2})$ where $E_{1} = \{ x \in X : g(x) \in B_{1} \}$ and $E_{2} = \{ x \in X : g(x) \in B_{2} \}$. Thus $\alpha \llbracket f \circ \beta, g \circ \beta \rrbracket = \big( (A_{1}, A_{2}) \wedge (D_{1}, D_{2}) \big) \vee \big( \neg (A_{1}, A_{2}) \wedge (E_{1}, E_{2}) \big)$. \\
  This evaluates to
  \begin{align*}
  \alpha \llbracket f \circ \beta, g \circ \beta \rrbracket
   & = \big( A_{1} \cap D_{1}, A_{2} \cup (A_{1} \cap D_{2}) \big) \vee \big( A_{2} \cap E_{1}, A_{1} \cup (A_{2} \cap E_{2}) \big) \\
   & = \Big( (A_{1} \cap D_{1}) \cup \big( (A_{2} \cup (A_{1} \cap D_{2})) \cap (A_{2} \cap E_{1})  \big), \\ & (A_{2} \cup (A_{1} \cap D_{2})) \cap (A_{1} \cup (A_{2} \cap E_{2})) \Big) \\
   & = (S_{1}, S_{2}) \text{   (say)}
  \end{align*}
  We show that $(C_{1}, C_{2}) = (S_{1}, S_{2})$ by standard set theoretic arguments.

  First we prove that $C_{1} \subseteq S_{1}$. Let $x \in C_{1}$. Then $h(x) \in B_{1}$. Consider the following cases: \\
    \emph{Case I}: $x \in A_{1}$: Then $h(x) = f(x) \in B_{1}$ hence $x \in D_{1}$. Therefore $x \in A_{1} \cap D_{1}$ and so $x \in S_{1}$. \\
    \emph{Case II}: $x \in A_{2}$: Then $h(x) = g(x) \in B_{1}$ hence $x \in E_{1}$. Hence $x \in A_{2} \cap E_{1} \subseteq A_{2}$ we have $x \in S_{1}$. \\
    \emph{Case III}: $x \notin (A_{1} \cup A_{2})$: Then $h(x) = \bot \notin B_{1}$ a contradiction to our assumption that $h(x) \in B_{1}$. It follows that this case cannot occur. \\

  We show that $S_{1} \subseteq C_{1}$. Let $x \in S_{1}$. Thus $x \in A_{1} \cap D_{1}$ or $x \in \big( (A_{2} \cup (A_{1} \cap D_{2})) \cap (A_{2} \cap E_{1})  \big)$. If $x \in A_{1} \cap D_{1}$ then $h(x) = f(x)$ as $x \in A_{1}$ and $f(x) \in B_{1}$ as $x \in D_{1}$. Thus $h(x) \in B_{1}$ and so $x \in C_{1}$. If $x \in \big( (A_{2} \cup (A_{1} \cap D_{2})) \cap (A_{2} \cap E_{1})  \big)$, then $x \in (A_{2} \cap E_{1})$. Thus $h(x) = g(x)$ as $x \in A_{2}$ and $g(x) \in B_{1}$ as $x \in E_{1}$. Hence $h(x) \in B_{1}$, thus $x \in C_{1}$. \\

   We show that $C_{2} \subseteq S_{2}$. Let $x \in C_{2}$ hence $h(x) \in B_{2}$. Consider the following cases: \\
    \emph{Case I}: $x \in A_{1}$: Then $h(x) = f(x) \in B_{2}$, therefore $x \in D_{2}$. Hence $x \in A_{1} \cap D_{2} \subseteq A_{1}$ and so $x \in S_{2}$. \\
    \emph{Case II}: $x \in A_{2}$: Then $h(x) = g(x) \in B_{2}$ therefore $x \in E_{2}$. Thus $x \in A_{2} \cap E_{2} \subseteq A_{2}$ and so $x \in S_{2}$. \\
    \emph{Case III}: $x \notin (A_{1} \cup A_{2})$: Then $h(x) = \bot \notin B_{2}$ which is a contradiction. It follows that this case cannot occur. \\

   Finally we show that $S_{2} \subseteq C_{2}$. Since $A_{1} \cap A_{2} = \emptyset$ it follows that $x \in A_{1} \cap D_{2}$ or $x \in A_{2} \cap E_{2}$. If $x \in A_{1} \cap D_{2}$ then $h(x) = f(x) \in B_{2}$ and hence $x \in C_{2}$. If $x \in A_{2} \cap E_{2}$ then $h(x) = g(x) \in B_{2}$ hence $x \in C_{2}$.

   Thus $\alpha[f, g] \circ \beta = \alpha \llbracket f \circ \beta, g \circ \beta \rrbracket$. \newline

\end{flushleft}

\subsection{Verification of \eref{ExampleGeneralCmonoid}} \label{VerGeneralCmonoid}

\begin{flushleft}

Let $f, g, h \in S_{\bot}^{X}$ and $\alpha, \beta \in \mathbb{3}^{X}$.

Axiom \eqref{EM7}: It is easy to see that $(\zeta_{\bot} \circ \alpha)(x) = U$ for all $x \in X$. \newline

Axiom \eqref{EM8}: It is clear that $(f \circ {\bf U})(x) = U$. \newline

Axiom \eqref{EM1}: Since $S_{\bot}$ is non-trivial we must have $1 \neq \bot$. If not then for $a \in S_{\bot} \setminus \{ \bot \}$ we have $a = a \cdot 1 = a \cdot \bot = \bot$ a contradiction. It follows that $\zeta_{1} \neq \zeta_{\bot}$. Hence $(\zeta_{1} \circ \alpha)(x) = \alpha(x)$ as $\zeta_{1}(x) = 1 \neq \bot$. \newline

Axiom \eqref{EM4}: We have
 \begin{align*}
  (f \circ (\neg \alpha))(x) & = \begin{cases}
                                 (\neg \alpha)(x), & \text{ if } f(x) \neq \bot; \\
                                 U, & \text{ otherwise} \\
                                 \end{cases} \\
                             & = \begin{cases}
                                 \neg (\alpha(x)), & \text{ if } f(x) \neq \bot; \\
                                 U, & \text{ otherwise} \\
                                 \end{cases} \\
                             & = \neg (f \circ \alpha)(x).
  \end{align*}
 Thus $f \circ (\neg \alpha) = \neg (f \circ \alpha)$. \newline

 Axiom \eqref{EM3}: We have
 \begin{align*}
  (f \circ (\alpha \wedge \beta))(x) & = \begin{cases}
                                         (\alpha \wedge \beta)(x), & \text{ if } f(x) \neq \bot; \\
                                         U, & \text{ otherwise} \\
                                         \end{cases} \\
                                     & = \begin{cases}
                                         \alpha(x) \wedge \beta(x), & \text{ if } f(x) \neq \bot; \\
                                         U \wedge U, & \text{ otherwise} \\
                                         \end{cases} \\
                                     & = (f \circ \alpha)(x) \wedge (f \circ \beta)(x).
 \end{align*}
 Thus $f \circ (\alpha \wedge \beta) = (f \circ \alpha) \wedge (f \circ \beta)$. \newline

 Axiom \eqref{EM2}: Since $S_{\bot}$ has no zero-divisors we have $f(x) \cdot g(x) = \bot \Leftrightarrow f(x) = \bot \text{ or } g(x) = \bot$. Consequently
 \begin{align*}
  ((f \cdot g) \circ \alpha)(x) & = \begin{cases}
                              \alpha(x), & \text{ if } (f \cdot g)(x) \neq \bot; \\
                              U, & \text{ otherwise} \\
                              \end{cases} \\
                          & = \begin{cases}
                              \alpha(x), & \text{ if } f(x) \cdot g(x) \neq \bot; \\
                              U, & \text{ otherwise} \\
                              \end{cases} \\
                          & = \begin{cases}
                              \alpha(x), & \text{ if } f(x) \neq \bot \text{ and } g(x) \neq \bot; \\
                              U, & \text{ otherwise} \\
                              \end{cases} \\
                          & = (f \circ (g \circ \alpha))(x).
 \end{align*}
 Thus $(f \cdot g) \circ \alpha = f \circ (g \circ \alpha)$. \newline

 Axiom \eqref{EM5}: We have
 \begin{align*}
  (\alpha[f, g] \cdot h)(x) = \alpha[f, g](x) \cdot h(x) & = \begin{cases}
                                                  f(x) \cdot h(x), & \text{ if } \alpha(x) = T; \\
                                                  g(x) \cdot h(x), & \text{ if } \alpha(x) = F; \\
                                                  \bot, & \text{ otherwise} \\
                                                 \end{cases} \\
                                               & = \alpha[f \cdot h, g \cdot h](x).
 \end{align*}
 Thus $\alpha[f, g] \cdot h = \alpha[f \cdot h, g \cdot h]$. \newline

 Axiom \eqref{EM9}: Consider
 \begin{equation*}
  h \cdot \alpha[f, g](x) = h(x) \cdot \alpha[f, g](x) = \begin{cases}
                                                h(x) \cdot f(x), & \text{ if } \alpha(x) = T; \\
                                                h(x) \cdot g(x), & \text{ if } \alpha(x) = F; \\
                                                \bot, & \text{ otherwise.}
                                               \end{cases}
 \end{equation*}

 On the other hand
 \begin{equation*}
  (h \circ \alpha)[h \cdot f, h \cdot g](x) = \begin{cases}
                                     h(x) \cdot f(x), & \text{ if } (h \circ \alpha)(x) = T; \\
                                     h(x) \cdot g(x), & \text{ if } (h \circ \alpha)(x) = F; \\
                                     \bot, & \text{ otherwise.}
                                    \end{cases}
 \end{equation*}

 Note that if $h(x) = \bot$ then $h \cdot \alpha[f, g](x) = \bot = (h \circ \alpha)[h \cdot f, h \cdot g](x)$. Suppose that $h(x) \neq \bot$ then $(h \circ \alpha)(x) = \alpha(x)$. It is clear that in this case as well $h \cdot \alpha[f, g](x) = (h \circ \alpha)[h \cdot f, h \cdot g](x)$ holds. Thus $h \cdot \alpha[f, g] = (h \circ \alpha)[h \cdot f, h \cdot g]$. \newline

  Axiom \eqref{EM6}: Consider
 \begin{align*}
  (\alpha[f, g] \circ \beta)(x) & = \begin{cases}
                                     \beta(x), & \text{ if } \alpha[f, g](x) \neq \bot; \\
                                     U, & \text{ otherwise} \\
                                    \end{cases} \\
                                & = \begin{cases}
                                     \beta(x), & \text{ if } (f(x) \neq \bot, \alpha(x) = T) \text{ or } (g(x) \neq \bot, \alpha(x) = F); \\
                                     U, & \text{ otherwise.} \\
                                    \end{cases}
 \end{align*}
 We have $(\alpha \llbracket f \circ \beta, g \circ \beta \rrbracket)(x) = (\alpha(x) \wedge (f \circ \beta)(x)) \vee (\neg \alpha(x) \wedge (g \circ \beta)(x))$.

 If $f(x) \neq \bot$ and $\alpha(x) = T$ we have $(\alpha \llbracket f \circ \beta, g \circ \beta \rrbracket)(x) = (T \wedge \beta(x)) \vee (F \wedge (g \circ \beta)(x)) = \beta(x) \vee F = \beta(x) = (\alpha[f, g] \circ \beta)(x)$.

 If $g(x) \neq \bot$ and $\alpha(x) = F$ we have $(\alpha \llbracket f \circ \beta, g \circ \beta \rrbracket)(x) = (F \wedge (f \circ \beta)(x)) \vee (T \wedge \beta(x)) = F \vee \beta(x) = \beta(x) = (\alpha[f, g] \circ \beta)(x)$.

 In all other cases it can be easily ascertained that $(\alpha \llbracket f \circ \beta, g \circ \beta \rrbracket)(x) = U = (\alpha[f, g] \circ \beta)(x)$. Thus $\alpha[f, g] \circ \beta = \alpha \llbracket f \circ \beta, g \circ \beta \rrbracket$. \newline

\end{flushleft}

\subsection{Verification of \eref{ExampleBasicCmonoid}} \label{VerBasicCmonoid}

\begin{flushleft}

Axiom \eqref{EM7}: It is clear that $\bot \circ \alpha = U$. \newline

Axiom \eqref{EM8}: It is obvious that $t \circ U = U$. \newline

Axiom \eqref{EM1}: Since $S_{\bot}$ is non-trivial it follows that $1 \neq \bot$. Consequently $1 \circ \alpha = \alpha$. \newline

Axiom \eqref{EM4}: If $s = \bot$ then $s \circ (\neg \alpha) = U = \neg ( s \circ \alpha)$. If $s \neq \bot$ then $s \circ (\neg \alpha) = \neg \alpha = \neg (s \circ \alpha)$. Thus $s \circ (\neg \alpha) = \neg (s \circ \alpha)$. \newline

Axiom \eqref{EM3}: If $s = \bot$ then $s \circ (\alpha \wedge \beta) = U$ and $(s \circ \alpha) \wedge (s \circ \beta) = U \wedge U = U$. If $s \neq \bot$ then $s \circ (\alpha \wedge \beta) = \alpha \wedge \beta = (s \circ \alpha) \wedge (s \circ \beta)$. Thus $s \circ (\alpha \wedge \beta) = (s \circ \alpha) \wedge (s \circ \beta)$. \newline

Axiom \eqref{EM2}: Consider $s, t \in S_{\bot}$ such that $s \cdot t = \bot$. Then $(s \cdot t) \circ \alpha = \bot \circ \alpha = U$. Since $S_{\bot}$ has no non-zero zero-divisors we have $s = \bot$ or $t = \bot$ and so $s \circ (t \circ \alpha) = U$ in either case. If $s \cdot t \neq \bot$ then $(s \cdot t) \circ \alpha = \alpha$ and $s \circ (t \circ \alpha) = t \circ \alpha = \alpha$ as neither $s$ nor $t$ are $\bot$. Thus $(s \cdot t) \circ \alpha = s \circ (t \circ \alpha)$. \newline

Axiom \eqref{EM5}: As $\alpha \in \{ T, F, U \}$ we consider the following three cases: \\
 \emph{Case I}: $\alpha = T$: Then $\alpha[s, t] \cdot u = T[s, t] \cdot u = s \cdot u = T[s \cdot u, t \cdot u]$. \\
 \emph{Case II}: $\alpha = F$: Then $\alpha[s, t] \cdot u = F[s, t] \cdot u = t \cdot u = F[s \cdot u, t \cdot u]$. \\
 \emph{Case III}: $\alpha = U$: Then $\alpha[s, t] \cdot u = U[s, t] \cdot u = \bot \cdot u = \bot = U[s \cdot u, t \cdot u]$. \\
 Thus $\alpha[s, t] \cdot u = \alpha[s \cdot u, t \cdot u]$. \newline

Axiom \eqref{EM9}: Consider the following cases: \\
 \emph{Case I}: $r = \bot$: Then $r \cdot \alpha[s, t] = \bot \cdot \alpha[s, t] = \bot = U[r \cdot s, r \cdot t] = (\bot \circ \alpha)[r \cdot s, r \cdot t] = (r \circ \alpha)[r \cdot s, r \cdot t]$. \\
 \emph{Case II}: $r \neq \bot$: We again consider the following three cases: \\
   \emph{Case i}: $\alpha = T$: $r \cdot \alpha[s, t] = r \cdot T[s, t] = r \cdot s = T[r \cdot s, r \cdot t] = (r \circ T)[r \cdot s, r \cdot t] = (r \circ \alpha)[r \cdot s, r \cdot t]$. \\
   \emph{Case ii}: $\alpha = F$: $r \cdot \alpha[s, t] = r \cdot F[s, t] = r \cdot t = F[r \cdot s, r \cdot t] = (r \circ F)[r \cdot s, r \cdot t] = (r \circ \alpha)[r \cdot s, r \cdot t]$. \\
   \emph{Case iii}: $\alpha = U$: $r \cdot \alpha[s, t] = r \cdot U[s, t] = r \cdot \bot = \bot = U[r \cdot s, r \cdot t] = (r \circ U)[r \cdot s, r \cdot t] = (r \circ \alpha)[r \cdot s, r \cdot t]$. \\
 Thus $r \cdot \alpha[s, t] = (r \circ \alpha)[r \cdot s, r \cdot t]$. \newline

 Axiom \eqref{EM6}: Consider the following three cases: \\
 \emph{Case I}: $\alpha = T$: $\alpha[s, t] \circ \beta = T[s, t] \circ \beta = s \circ \beta = T \llbracket s \circ \beta, t \circ \beta \rrbracket$. \\
 \emph{Case II}: $\alpha = F$: $\alpha[s, t] \circ \beta = F[s, t] \circ \beta = t \circ \beta = F \llbracket s \circ \beta, t \circ \beta \rrbracket$. \\
 \emph{Case III}: $\alpha = U$: $\alpha[s, t] \circ \beta = U[s, t] \circ \beta = \bot \circ \beta = U = U \llbracket s \circ \beta, t \circ \beta \rrbracket$. \\
 Thus $\alpha[s, t] \circ \beta = \alpha \llbracket s \circ \beta, t \circ \beta \rrbracket$. \newline

\end{flushleft}

\end{document}